\documentclass[11pt,a4paper]{article}

\usepackage[T1]{fontenc}
\usepackage[utf8]{inputenc}

\usepackage{listings}
\usepackage{todonotes}
\usepackage{multicol}
\usepackage{amssymb}
\usepackage{amsthm}
\usepackage{mathtools}

\usepackage{amsfonts}
\usepackage{enumerate}
\usepackage{microtype}
\usepackage{paralist}
\usepackage{algorithm}
\usepackage{algpseudocode}

\algrenewcommand\algorithmicindent{0.7em}

\usepackage{graphicx}
\usepackage{subcaption}

\usepackage{mathtools}
\usepackage{amsthm}
\usepackage{pgfplots}
\usepackage[inline]{enumitem}
\usepackage{cleveref}
\usepackage{authblk}

\theoremstyle{definition}
\newtheorem{definition}{Definition}[section]

\newtheorem{remark}[definition]{Remark}
\newtheorem{claim}[definition]{Claim}
\newtheorem{requirement}[definition]{Requirement}

\theoremstyle{plain}
\newtheorem{corollary}[definition]{Corollary}
\newtheorem{lemma}[definition]{Lemma}

\newtheorem{theorem}[definition]{Theorem}

\numberwithin{equation}{section}

\DeclarePairedDelimiter{\abs}{\lvert}{\rvert}

\DeclarePairedDelimiter{\parens}{\lparen}{\rparen}
\DeclarePairedDelimiter{\bracks}{\lbrack}{\rbrack}
\DeclarePairedDelimiter{\braces}{\{}{\}}

\newcommand{\algid}[1] {$#1$}

\newcommand{\euler}{e}

\newcommand{\wsslong}{Work Stealing and Spreading}
\newcommand{\wss}{WSS}

\newcommand{\wslong}{Work Stealing}
\newcommand{\ws}{WS}

\newcommand{\iexecd}{\iota}

\algblockdefx[atomically]{BeginAtomic}{EndAtomic}{\textbf{atomically do}}{\textbf{done}}

\title{On the analysis of scheduling algorithms for structured parallel computations}

\date{\today}

\author[1]{Guilherme Rito\thanks{Work done while author was at NOVA Laboratory for Computer Science and Informatics,  Departamento de Informática, Faculdade de Ciências e Tecnologia, Universidade NOVA de Lisboa.}}

\author[2]{Hervé Paulino}

\affil[1]{ETH Zurich, guilherme.teixeira@inf.ethz.ch}

\affil[2]{NOVA Laboratory for Computer Science and Informatics,  Departamento de Informática, Faculdade de Ciências e Tecnologia, Universidade NOVA de Lisboa, herve.paulino@fct.unl.pt}

\begin{document}
\maketitle

\begin{abstract}
  Algorithms for scheduling structured parallel computations have been widely studied in the literature.
  For some time now, \wslong\ is one of the most popular for scheduling such computations, and its performance has been studied in both theory and practice.
  Although it delivers provably good performances, the effectiveness of its underlying load balancing strategy is known to be limited for certain classes of computations, particularly the ones exhibiting irregular parallelism (\textit{e.g.} depth first searches).
  Many studies have addressed this limitation from a purely load balancing perspective, viewing computations as sets of independent tasks, and then analyzing the expected amount of work attached to each processor as the execution progresses.
  However, these studies make strong assumptions regarding work generation which, despite being standard from a queuing theory perspective --- where work generation can be assumed to follow some random distribution --- do not match the reality of structured parallel computations --- where the work generation is not random, only depending on the structure of a computation.

  In this paper, we introduce a formal framework for studying the performance of structured computation schedulers, define a criterion that is appropriate for measuring their performance, and present a methodology for analyzing the performance of randomized schedulers.
  We demonstrate the convenience of this methodology by using it to prove that the performance of \wslong\ is limited, and to analyze the performance of a \wsslong\ algorithm, which overcomes \wslong's limitation.
\end{abstract}

\section{Introduction}
\label{sec:intro}
The main goal of a structured computation's scheduler is to guarantee the fast completion of the execution of arbitrary structured computations.
For some time now, \wslong\ is one of the most popular algorithms for scheduling structured computations~\cite{DBLP:conf/spaa/AroraBP98,DBLP:journals/mst/AroraBP01,DBLP:journals/jacm/BlumofeL99,DBLP:conf/ppopp/AcarCR13,DBLP:journals/siamcomp/BerenbrinkFG03,DBLP:conf/icpp/CongKKLSW08,DBLP:conf/sc/DinanLSKN09,DBLP:conf/podc/HendlerS02,DBLP:conf/europar/QuintinW10}.
In \wslong\ (or \ws\ for short), each processor owns a deque that uses to keep track of its work.
Busy processors operate locally on their deques, adding and retrieving work from them as necessary, until they run out of work.
When that happens, a processor becomes a thief and starts a stealing phase, during which it targets other processors, uniformly at random, in order to steal work from their deques.

As proved in~\cite{DBLP:journals/mst/AroraBP01,DBLP:journals/jacm/BlumofeL99}, the expected execution time of any computation using \ws\ is asymptotically optimal.
Nevertheless, \ws's performance is known to be limited for the execution of computations that exhibit irregular parallelism (\textit{e.g.}~depth first computations where only a few threads actually generate work)~\cite{DBLP:conf/ppopp/AcarCR13,DBLP:journals/siamcomp/BerenbrinkFG03,DBLP:conf/icpp/CongKKLSW08,DBLP:conf/sc/DinanLSKN09}.
For coping with this limitation, numerous studies have been resorting to the use of steal-half deques~\cite{DBLP:conf/ppopp/AcarCR13,DBLP:conf/sc/DinanLSKN09,DBLP:conf/podc/HendlerS02,DBLP:conf/europar/QuintinW10} which allow thieves to take up to half of the work of their victims in a single steal operation.
The adoption of steal-half strategies by real-life schedulers has been mostly justified by the strategies' importance on distributed memory environments, where each steal attempt incurs in significant latency, making it worth to transfer a larger amount of work in a single steal.
On the other hand, the steal-half strategy has been formally proved, from a queuing theory perspective, to be an effective load balancing method for schedulers of independent tasks~\cite{DBLP:journals/siamcomp/BerenbrinkFG03}.
However, while this strategy may be ideal for independent task scheduling from a queuing theory perspective --- where tasks are assumed to arrive at a system according to some probability distribution, and work transfers are assumed to take constant time regardless of the amount of tasks transferred --- it remains unknown whether it is suitable for structured computation scheduling --- where work generation only depends on the structure of a computation, and where the time for a processor to transfer work from another processor is proportional to the amount of work transferred --- and so the problem of how to cope with \ws's limitation remains open.
Even more importantly, while there are well established methods for analyzing the performance of the load balancers of independent task schedulers --- usually based on the analysis of Markov chains --- and a well-defined goal --- which is typically to assure that the system's load does not grow unboundedly over time --- to the best of our knowledge there are no well-defined methods suitable for analyzing the performance of the load balancers of online structured computation schedulers, nor even well-defined goals.


To this extent, the contributions of this paper are:
\begin{itemize}[noitemsep,topsep=0pt,parsep=0pt,partopsep=0pt]
\item A formal framework for studying the performance of structured computation schedulers (Section~\ref{sec:preliminaries}).
  One of the key features of this framework is that it can be used to model most, if not all, practical scheduling algorithms.

\item The definition of \textit{algorithm short-term stability} (Section~\ref{subsec:stt}), which is an appropriate criterion for measuring the performance of online structured computation schedulers.
    
\item A methodology that allows to effectively study the performance of randomized computation schedulers (Section~\ref{sec:strategy-randomized-analysis}).
  We demonstrate its convenience by:
  \begin{enumerate*}
  \item using it to prove that the performance of \ws\ is indeed limited (Section~\ref{subsec:identifying-ws-limitation}); and
  \item presenting a (purely theoretical) variant of the \ws\ algorithm where processors attempt to spread work as it is generated, and then using our methodology to show that the algorithm overcomes the identified limitations of \ws\ (Section~\ref{sec:analyzing-wss}).
  \end{enumerate*}
  Despite being purely theoretical, the algorithm we present gives us insight on how the limitations of \ws\ can be addressed.
\end{itemize}
\section{A criterion to measure the performance of computation schedulers}
\label{sec:preliminaries}




Like in much previous work~\cite{DBLP:journals/mst/AcarBB02,DBLP:journals/tocs/AgrawalLHH08,DBLP:journals/mst/AroraBP01,DBLP:journals/jacm/BlumofeL99,DBLP:conf/spaa/MullerA16,DBLP:conf/isaac/TchiboukdjianGTRB10}, we model a computation as a \emph{dag} $G = \left (V, E \right)$, where each node $v \in V$ corresponds to an instruction, and each edge $\left(\mu_{1},\mu_{2}\right) \in E$ denotes an ordering constraint (meaning $\mu_{2}$ can only be executed after $\mu_{1}$).
Nodes with in-degree of 0 are referred to as \emph{roots}, while nodes with out-degree of 0 are called \emph{sinks}.
We make the two standard assumptions related with the structure of computations.
Let $G$ denote a computation's dag:
\begin{enumerate*}
\item there is only one root and one sink in $G$; and
\item the out-degree of any node within $G$ is at most two.
\end{enumerate*}

We consider that processors operate on discrete time steps, each executing one instruction --- that may or may not correspond to a computation node --- per time step.
The execution of a computation is carried out by a set of processors denoted by $Procs$ whose cardinality is denoted by $P$.
We assume that $P \geq 2$ (\textit{i.e.}~$Procs$ is composed by at least two processors), and that all processors operate synchronously in time steps.
Therefore, a computation's execution can be partitioned into discrete time steps, such that at each step every processor executes an instruction.
We refer to these time steps using non-negative integers, where $0$ is the first step and $i+1$ is the step succeeding $i$.


\begin{definition}
  \label{def:node-states}
  At any step during a computation's execution each node of the computation is in exactly one of the following states:
  \begin{itemize*}
  \item[\textbf{not ready}] --- if its predecessors have not yet been executed; 
  \item[\textbf{ready}] --- if its predecessors have been executed, but not the node itself; and
  \item[\textbf{executed}] --- if the node has been executed.
  \end{itemize*}
\end{definition}

As one may note, a node can only be \textbf{ready} if all the ordering constraints wrt (with respect to) the node are satisfied.
For example, at the first step of a computation's execution every node (except for the root) is \textbf{not ready}.
To ensure the correct execution of a computation, only nodes that are \textbf{not ready} can become \textbf{ready}, and only nodes that are \textbf{ready} can become \textbf{executed}.
For each step $i$, refer to the set of nodes that are:
\begin{enumerate*}
\item \textbf{not ready} by $NonReady_{i}$;
\item \textbf{ready} by $Ready_{i}$ (or simply $R_{i}$); and
\item \textbf{executed} by $Executed_{i}$.
\end{enumerate*}
Since only the nodes that are \textbf{ready} can become \textbf{executed}, $Executed_{i}$ can alternatively be defined as $Executed_{i} = \parens[\big]{\bigcup_{j \in \braces{1,\ldots,i-1}}R_{j}} - R_{i}$ (\textit{i.e.} the set of all nodes that were once ready, but no longer are).

For each step $i$, partition $R_{i}$ into $P$ sets (one per processor), and refer to $R_{i}\parens{p}$ --- processor $p$'s partition of $R_{i}$ ---  as the set of nodes that are \emph{attached} to processor $p$ at step $i$.
Say that a node was \emph{enabled} at step $i$ if it was \textbf{not ready} at step $i$ but is \textbf{ready} at step $i+1$, and, similarly, that a node was \emph{executed} at step $i$ if it was \textbf{ready} at step $i$ but is \textbf{executed} at step $i+1$.
In addition, say that a node $\mu$ is \emph{migrated} if $\mu \in R_{i}\left(p\right)$ and $\mu \in R_{i+1}\left(q\right)$, where $p \neq q$, for $p,q \in Procs$.
The next definition formalizes these ideas.


\begin{definition}
  \label{def:node-definitions}
  For each step $i$ and processor $p$, define the set of nodes \emph{enabled} by $p$ as $E_{i}\parens{p} = R_{i+1}\parens{p} - R_{i}$ and \emph{computed} (or \emph{executed}) by $p$ as $C_{i}\parens{p} = R_{i}\parens{p} - R_{i+1}$.
  Moreover, define the set of nodes \emph{migrated} from $p$ to all other processors as $M^{+}_{i}\parens{p} = R_{i}\parens{p} \cap \parens{R_{i+1} - R_{i+1}\parens{p}}$ and from all other processors to $p$ as $M^{-}_{i}\parens{p} = R_{i+1}\parens{p} \cap \parens{R_{i} - R_{i}\parens{p}}$.

  For a set of processors $S \in \mathcal{P} \parens{Procs}$, define $R_{i}\parens{S} = \bigcup_{p \in S} R_{i}\parens{p}$, $E_{i}\parens{S} = \bigcup_{p \in S} E_{i}\parens{p}$, $C_{i}\parens{S} = \bigcup_{p \in S} C_{i}\parens{p}$, $M^{+}_{i}\parens{S} = \bigcup_{p \in S} M^{+}_{i}\parens{p}$, and $M^{-}_{i}\parens{S} = \bigcup_{p \in S} M^{-}_{i}\parens{p}$.
  Additionally, define $E_{i} = E_{i}\parens{Procs}$, $C_{i} = C_{i}\parens{Procs}$, $M^{+}_{i} = M^{+}_{i}\parens{Procs}$, and $M^{-}_{i} = M^{-}_{i}\parens{Procs}$.
\end{definition}




Having defined these sets of nodes, we now introduce rounds.
Informally, a round is a sequence of time steps with constant length, such that every processor executes at most one node and no ready node is migrated more than once.

\begin{definition}
  \label{def:round}
  A \emph{round} is a sequence of $L$ time steps (for some constant $L \geq 1$) such that a computation's execution can be partitioned into equal-length rounds and for every round:
  \begin{enumerate*}
  \item no processor executes more than a single node; and 
  \item no node is migrated more than once.
  \end{enumerate*}
\end{definition}

Analogously to time steps, refer to rounds using non-negative integers, but with an additional bar, where $\overline{0}$ denotes the first round.
Throughout this paper, we let $L$ denote the length of rounds and $\overline{t}\left[i\right]$ denote the $i$-\textit{th} step of a round $\overline{t}$ (for $i \in \braces{0,\ldots,L-1}$).
As we will see, the length of rounds depends on the scheduling algorithm.


Now, we reintroduce the concepts we have already presented above concerning the states of nodes, but this time considering rounds rather than steps.

\begin{definition}
  \label{def:node-definitions-round}
  For each round $\overline{t}$ and processor $p$, define the set of nodes \emph{attached} to $p$ at round $\overline{t}$ as $R_{\overline{t}}\parens{p} = R_{\overline{t}\bracks{0}}\parens{p}$, \emph{enabled} by $p$ during $\overline{t}$ as \[E_{\overline{t}}\parens{p} = \bigcup_{i \in \braces{\overline{t}\bracks{0},\cdots,\overline{t}\bracks{L - 1}}}E_{i}\parens{p},\] \emph{executed} by $p$ during $\overline{t}$ as \[C_{\overline{t}}\parens{p} = \bigcup_{i \in \braces{\overline{t}\bracks{0},\cdots,\overline{t}\bracks{L - 1}}}C_{i}\parens{p},\] \emph{migrated} from $p$ to all other processors during $\overline{t}$ as \[M^{+}_{\overline{t}}\parens{p} = \bigcup_{i \in \braces{\overline{t}\bracks{0},\cdots,\overline{t}\bracks{L - 1}}}M^{+}_{i}\parens{p}\] and \emph{migrated} from all other processors to $p$ during $\overline{t}$ as \[M^{-}_{\overline{t}}\parens{p} = \bigcup_{i \in \braces{\overline{t}\bracks{0},\cdots,\overline{t}\bracks{L - 1}}}M^{-}_{i}\parens{p}.\]

  For a set of processors $S \in \mathcal{P} \parens{Procs}$, define $R_{\overline{t}}\parens{S} = \bigcup_{p \in S} R_{\overline{t}}\parens{p}$, $E_{\overline{t}}\parens{S} = \bigcup_{p \in S} E_{\overline{t}}\parens{p}$, $C_{\overline{t}}\parens{S} = \bigcup_{p \in S} C_{\overline{t}}\parens{p}$, $M^{+}_{\overline{t}}\parens{S} = \bigcup_{p \in S} M^{+}_{\overline{t}}\parens{p}$, and $M^{-}_{\overline{t}}\parens{S} = \bigcup_{p \in S} M^{-}_{\overline{t}}\parens{p}$.
  Additionally, define $E_{\overline{t}} = E_{\overline{t}}\parens{Procs}$, $C_{\overline{t}} = C_{\overline{t}}\parens{Procs}$, $M^{+}_{\overline{t}} = M^{+}_{\overline{t}}\parens{Procs}$, and $M^{-}_{\overline{t}} = M^{-}_{\overline{t}}\parens{Procs}$.
\end{definition}

The following result, Lemma~\ref{lemma:migrated-to-p-subset-migrated-from-all-but-p}, states that the set of nodes migrated to a processor $p$ during a round $\overline{t}$ is a subset of all the nodes that are migrated from all processors but $p$ during $\overline{t}$.
The proof of the following result can be found in the Appendix (Section~\ref{sub:1-proof-lemma-1}).

\begin{lemma}
  \label{lemma:migrated-to-p-subset-migrated-from-all-but-p}
  For any round $\overline{t}$ and $p \in Procs$, $M^{-}_{\overline{t}}\parens{p} \subseteq M^{+}_{\overline{t}}\parens{Procs - \braces{p}}$.
\end{lemma}

As one may note, the definition of round (Definition~\ref{def:round}) implies that for any processor $p$ and round $\overline{t}$, $\abs{C_{\overline{t}}\parens{p}} \leq 1$ and $M^{+}_{\overline{t}}\parens{p} \cap M^{+}_{\overline{t}}\parens{Procs - \braces{p}} = \emptyset$ (\textit{i.e.}~no two processors migrate the same node during the same round).
By Lemma~\ref{lemma:migrated-to-p-subset-migrated-from-all-but-p}, it then follows $M^{+}_{\overline{t}}\parens{p} \cap M^{-}_{\overline{t}}\parens{p} = \emptyset$.



The next lemma is essential for the rest of our analysis, as it shows the connection between the set of nodes that are attached to each processor $p$ at some round $\overline{t}$, and the set of nodes that are attached to $p$ at round $\overline{t+1}$.
%
%
The proof of this result can be found in the Appendix (Section~\ref{sub:1-proof-lemma-2}).

\begin{lemma}[Round Progression Lemma]
\label{lemma:r-next}
For any round $\overline{t}$ and processor $p \in Procs$,
\begin{align*}
  R_{\overline{t+1}}\left(p\right) &= \left(E_{\overline{t}}\left(p\right) \cup R_{\overline{t}}\left(p\right) \cup M^{-}_{\overline{t}}\left(p\right)\right) - \left(C_{\overline{t}}\left(p\right) \cup M^{+}_{\overline{t}}\left(p\right)\right).
\end{align*}
\end{lemma}

\subsection{Algorithm short-term stability}
\label{subsec:stt}

We now move to present \emph{algorithm short-term stability}, the criterion that will be used to measure the performance of structured computation schedulers.
We begin by stating the following requirement, which gives us the guarantee that a processor $p$ only executes a node $\mu$ during a round $\overline{t}$ if $\mu$ is attached to $p$ at the beginning of that round.

\begin{requirement}
  \label{requirement:node-execution-constraints}
  For any round $\overline{t}$ and $p \in Procs$, we must have $R_{\overline{t}}\left(p\right) \supseteq C_{\overline{t}} \left(p\right)$.
\end{requirement}


\begin{definition}[Busy and Idle Processors]
  \label{def:processor-state-round}
  Say that a processor $p \in Procs$ is \emph{idle} during a round $\overline{t}$ if $C_{\overline{t}}\parens{p} = \emptyset$, and, otherwise, say that $p$ is \emph{busy}.
  Moreover, denote the number of idle processors during a round $\overline{t}$ by $P^{idle}_{\overline{t}}$, and define $\alpha_{\overline{t}}$ as the ratio of idle processors, $\alpha_{\overline{t}} = P^{idle}_{\overline{t}}/P$.
\end{definition}


Now, we introduce the notion of \emph{short-term stability}.
Intuitively, a set of processors $S$ is \emph{short-term stable} for some round $\overline{t}$ if the number of nodes attached to the processors in $S$ that are not executed is expected to \textbf{monotonically} decrease from round $\overline{t}$ to round $\overline{t+1}$.

\begin{definition}[Short-term stability]
  \label{def:short-term-stability}
  A set of processors $S \in \mathcal{P} \parens{Procs}$ is \emph{short-term stable} for some round $\overline{t}$ during a computation's execution if $\mathrm{E}\bracks{\abs{R_{\overline{t+1}}\parens{S} - C_{\overline{t+1}}\parens{S}}} \leq \abs{R_{\overline{t}}\parens{S} - C_{\overline{t}}\parens{S}}$.
\end{definition}


Ideally, we would want to ensure \emph{short-term stability} for all rounds and wrt all processors (\textit{i.e.}~$S = Procs$).
However, since a processor can enable two nodes during one round, a scheduler may only be able to guarantee \emph{short-term stability} wrt all processors if at least half of them are idle during a round.
For this reason, we will now move to introduce \emph{Algorithm short-term stability}, which is based on the same rationale as \emph{short-term stability}.


For each round $\overline{t}$, we classify processors according to whether they execute all their attached nodes during $\overline{t}$ or not.
If a processor $p$ executes all its attached nodes during round $\overline{t}$ (\textit{i.e.}~if $R_{\overline{t}}\left(p\right) = C_{\overline{t}}\left(p\right)$), then we say that $p$ is \emph{self-stable} at round $\overline{t}$.
Otherwise, we say that $p$ is \emph{non-self-stable} at round $\overline{t}$.

\begin{definition}
  \label{def:stable-unstable-processors}
  Define the set of \emph{self-stable} and \emph{non-self-stable} processors at some round $\overline{t}$ as $S_{\overline{t}} = \left\{p \in Procs \, | \,R_{\overline{t}}\left(p\right) = C_{\overline{t}}\left(p\right)\right\}$ and $U_{\overline{t}} = Procs - S_{\overline{t}}$, respectively.
\end{definition}


Having this, we can finally define \emph{Algorithm short-term stability}, our criterion for measuring computation schedulers' performance.
Informally, the main idea is that if the ratio of idle processors at some round $\overline{t}$ is sufficiently high, then the amount of work attached to non-self-stable processors is expected to decrease, and, at the same time, the amount work attached to self-stable processors does not grow unboundedly.


\begin{definition}[Algorithm short-term stability]
  \label{def:short-term-stable-algorithm}
  A scheduling algorithm is \emph{algorithm short-term stable} with respect to an interval $I \subseteq \left]0;1\right[$, iff (if and only if) for any round $\overline{t}$,
  \begin{eqnarray*}
    \left(\alpha_{\overline{t}} \in I\right) & \Rightarrow & \bracks[\bigg]{\parens[\Big]{\mathrm{E}\bracks{\abs{R_{\overline{t+1}}\parens{U_{\overline{t}}} - C_{\overline{t+1}}\parens{U_{\overline{t}}}}} < \abs{R_{\overline{t}}\parens{U_{\overline{t}}} - C_{\overline{t}}\parens{U_{\overline{t}}}}} \\
    & & \quad \wedge \ \parens[\Big]{\forall p \in S_{\overline{t}}, \left|R_{\overline{t+1}}\left(p\right)\right| \leq L + 1}},
  \end{eqnarray*}
  where $L$ denotes the length of the rounds.
\end{definition}

Note that, contrarily to \emph{short-term stability}, \emph{algorithm short-term stability} requires that the expected number of nodes attached to processors of $U_{\overline{t}}$, that are not executed, \textbf{strictly} decreases from a round to the next.
In addition, by limiting the number of nodes that can become attached to a self-stable processor during a round, we disallow scheduling algorithms to keep ping-ponging work between non-self-stable and self-stable processors throughout the execution.
The insight for bounding the number of nodes to the length of rounds is that we are enforcing processors to have to accept each node they are given.

Intuitively, if an algorithm is \emph{algorithm short-term stable} wrt some non-empty interval $I$, then the algorithm's load balancer is sufficiently effective to guarantee that, regardless of the computation it is scheduling, for any round $\overline{t}$ such that $\alpha_{\overline{t}} \in I$, its performance is expected to be good (or, in other words, work accumulation is not expected).
On the other hand, if an algorithm cannot guarantee \emph{algorithm short-term stability} wrt any non-empty interval $I$ for any round during the execution of an arbitrary computation, then the effectiveness of its load balancer is limited, and thus may lead to work accumulation depending on the structure of the computation being executed.
%




As one may note, the definition of \emph{Algorithm short-term stability} relies on the overall behavior of the set of processors $U_{\overline{t}}$, for each round $\overline{t}$.
However, it is much simpler to reason about the behavior of each processor $p \in U_{\overline{t}}$, than it is to reason about the behavior of all processors of $U_{\overline{t}}$.
The next result is then crucial for the rest of our analysis, as it shows the relation between the behavior of all the processors of a set of processors $S$ and the behavior of each individual processor $p$ of $S$.

\begin{lemma}
  \label{lemma:stability-generalize}
  For any round $\overline{t}$ and $S \in \mathcal{P} \left(Procs\right)$, if \[\forall p \in S ,\, \mathrm{E}\left[\left|R_{\overline{t+1}}\left(p\right) - C_{\overline{t+1}}\left(p\right)\right|\right] < \left|R_{\overline{t}}\left(p\right) - C_{\overline{t}}\left(p\right)\right|\]
  then \[\mathrm{E}\left[\left|R_{\overline{t+1}}\left( S\right) - C_{\overline{t+1}}\left( S\right)\right|\right] < \left|R_{\overline{t}}\left( S\right) - C_{\overline{t}}\left( S\right)\right|.\]
\end{lemma}
\begin{proof}
  First, recall that $R_{\overline{t}}$ is partitioned through the processors.
  By Requirement~\ref{requirement:node-execution-constraints}, $C_{\overline{t}}$ is also partitioned through the processors and $\forall p \in Procs$, $R_{\overline{t}}\left(p\right) \supseteq C_{\overline{t}} \left(p\right)$.
  Thus, \[\sum_{p \in S} \left|R_{\overline{t}}\left(p\right) - C_{\overline{t}}\left(p\right)\right| = \left|R_{\overline{t}}\left( S\right) - C_{\overline{t}}\left( S\right)\right|\]
  To conclude the proof, note that by the linearity of expectation, \[\sum_{p \in S}E\left[\left|R_{\overline{t+1}}\left(p\right) - C_{\overline{t+1}}\left(p\right)\right|\right] = E\left[\left|R_{\overline{t+1}}\left(S\right) - C_{\overline{t+1}}\left( S\right)\right|\right]\]
\end{proof}



The following result is base for the analysis of schedulers, as it relates, for each round $\overline{t}$, the difference in the number of nodes that a processor $p$ enables during $\overline{t}$, that are migrated from $p$ during $\overline{t}$, and that $p$ executes during $\overline{t+1}$, with a result that is closely related with \textit{algorithm short-term stability} (the corresponding proof can be found in the Appendix, Section~\ref{sub:1-proof-lemma-4}).

\begin{lemma}
\label{lemma:connecting-lemma}
For any round $\overline{t}$ and processor $p \in Procs$, \[\left|E_{\overline{t}}\left(p\right)\right| < \left|C_{\overline{t+1}}\left(p\right)\right| + \abs[\big]{M^{+}_{\overline{t}}\left(p\right)}\] iff \[\left|R_{\overline{t+1}}\left(p\right) - C_{\overline{t+1}}\left(p\right)\right| < \left|R_{\overline{t}}\left(p\right) - C_{\overline{t}}\left(p\right)\right| + \abs[\big]{M^{-}_{\overline{t}}\left(p\right)}.\]
\end{lemma}

The following Corollary then follows from Lemma~\ref{lemma:connecting-lemma}.

\begin{corollary}
\label{corollary:connecting-corollary}
For any round $\overline{t}$, if $\left(p \in U_{\overline{t}}\right) \Rightarrow \left(M^{-}_{\overline{t}}\left(p\right) = \emptyset\right)$, then
\[\left|E_{\overline{t}}\left(p\right)\right| < \left|C_{\overline{t+1}}\left(p\right)\right| + \abs[\big]{M^{+}_{\overline{t}}\left(p\right)}\] iff \[\left|R_{\overline{t+1}}\left(p\right) - C_{\overline{t+1}}\left(p\right)\right| < \left|R_{\overline{t}}\left(p\right) - C_{\overline{t}}\left(p\right)\right|.\]
\end{corollary}
\section{A method to analyze randomized schedulers}
\label{sec:strategy-randomized-analysis}


In order to analyze the performance of randomized scheduling algorithms, we introduce a few additional definitions and make some assumptions that are necessary to permit ordering the actions that processors take during the execution of computations, and, in particular, during each round.
The reason for the need to order the actions of processors will become apparent as we use it to analyze the \ws\ algorithm.
To aid the reader, as we present the extra definitions and assumptions that our methodology requires, we use a \ws\ algorithm (depicted in Algorithm~\ref{algo:ws}) to instantiate them and explain their meaning.

The \ws\ algorithm we analyze is a synchronous but behaviorally equivalent variant of the original non-blocking algorithm given in~\cite{DBLP:journals/mst/AroraBP01}.
Thus, each processor owns a lock-free deque object that supports three methods: \algid{pushBottom}, \algid{popBottom} and \algid{popTop}.
Only the owner of a deque may invoke the \algid{pushBottom} and \algid{popBottom} methods, which, respectively, add a node to the bottom of the deque, and remove and return the bottommost node of the deque, if any.
The \algid{popTop} method is invoked by processors searching for work, and for each invocation to this method, the deque's current topmost node is guaranteed to be removed and returned, either by such invocation or by some concurrent one\footnote{For a more careful description of the lock-free deque semantics, originally defined in~\cite{DBLP:journals/mst/AroraBP01}, please refer to Section~\ref{sec:deque-semantics}.}.
In addition to the deque, each processor has a variable \algid{assigned} that stores the node that it will execute next, if any.

\subsection{The methodology}
\label{subsec:methodology}


First of all, we require that the scheduling algorithm to be analyzed must be defined by a cycle such that:
\begin{enumerate}
\item at most one of the instructions composing any particular iteration of this cycle may correspond to a node's execution;
\item no node that is migrated to a processor $p$, who is executing an iteration of this cycle, can be migrated again (to another processor), before $p$ finishes the current iteration;
\item the length of any sequence of instructions that corresponds to some execution of this cycle is at most constant; and 
\item the full sequence of instructions executed by any processor can be partitioned into smaller sub-sequences, each corresponding to a particular execution of this cycle.
\end{enumerate}
Refer to this cycle as the \emph{scheduling loop}, and to any sequence of instructions that correspond to some iteration of a scheduling loop as \emph{scheduling iteration}.

As it can be observed in Algorithm~\ref{algo:ws}, the definition of the \ws\ algorithm naturally fits into scheduling loops (corresponding to lines 2 to 19):
\begin{enumerate*}
\item at most one of the instructions within the sequence of a scheduling iteration corresponds to the execution of a node (line 4);
\item no node that is migrated to a processor is migrated ever again, as it becomes the processor's new assigned node (line 23);
\item the length of any iteration of the scheduling loop is bounded by a constant; and
\item the full sequence of instructions executed by any processor can be partitioned into scheduling iterations.
\end{enumerate*}


  \begin{algorithm}[t]
    \caption{The synchronous \ws\ algorithm}
    \label{algo:ws}
    \begin{small}
      \begin{algorithmic}[1]
        \Procedure{Scheduler}{}
        \While{\textbf{not} finished ($computation$)}
        \If{ValidNode($assigned$)}
        \State $enabled \gets $execute($assigned$)\;
        \State $assigned \gets $ \textsc{none}\;
        \State synch($max\_phase_{I}\_length$, $\iexecd()$)\;
        \If{length($enabled$) $> 0$}
        \State $assigned \gets enabled\left[0\right]$\;
        \If{length($enabled$) $= 2$}
        \State $self.deque$.pushBottom($enabled\left[1\right]$)\;
        \EndIf
        \Else
        \State $assigned \gets self.deque$.popBottom()\;
        \EndIf
        \Else
        \State $self$.WorkMigration()\;
        \EndIf
        \State synch($max\_phase_{II}\_length$, $\iexecd()$)\;
        \EndWhile
        \EndProcedure
                \newline
        \Procedure{WorkMigration}{}
        \State $victim \gets $ UniformlyRandomProcessor()\;
        \State $assigned \gets victim.deque$.popTop()\;
        \State synch($max\_phase_{I}\_length$, $\iexecd()$)\;
        \EndProcedure
        \newline
        \Function{ValidNode}{$node$}
        \State \Return~$node\,\neq\,$\textsc{empty} \\  $\qquad and \quad node\,\neq\,$ \textsc{abort} \\ $\qquad and \quad node\,\neq\,$ \textsc{none}\;
        \EndFunction
      \end{algorithmic}
    \end{small}
  \end{algorithm}


To order the actions that processors take during scheduling iterations, each iteration can be partitioned into a sequence of \emph{phases}.
In particular, for \ws, each iteration is partitioned into two phases:
\begin{description}[leftmargin=0cm,noitemsep,topsep=5pt,parsep=5pt,partopsep=0pt]
\item[Phase I]
  If a processor has a valid assigned node, it executes the node.
  Otherwise, it makes a steal attempt, and if the attempt succeeds the stolen node becomes the processor's new assigned node.
  
\item[Phase II]
  If a processor made a steal attempt in the previous phase, it takes no action during this phase.
  Otherwise, the processor executed a node in the previous stage, which enabled either 0, 1 or 2 nodes.
  If no node was enabled, the processor invokes \algid{popBottom} to fetch the bottommost node from its deque, if such node exists.
  If at least one node was enabled, one of the enabled nodes becomes the processor's new assigned node, whist the other node, if any, is pushed by the processor into the bottom of its own deque, via the \algid{pushBottom} method.
\end{description}


At this point, we have already ordered the actions that each processor takes during the execution of every iteration.
However, this ordering by itself does not meet our needs, as we have to guarantee that all processors start the execution of each phase of every scheduling iteration at the same time.
Our first step towards that goal is to require all processors to begin working at the exact same time.
Refer to the step at which a processor $p$ executes its $i$-th instruction as $\chi\left(p,i\right)$.

\begin{requirement}
  \label{requirement:processors-start-synchronized}
  $\forall p \in Procs, \quad \chi\left(p,1\right) = 0$.
\end{requirement}


Now, we present the \algid{synch} procedure, which allows to synchronize processors at the end of each phase.
The \algid{synch} procedure takes two input parameters:
\begin{enumerate*}
\item $maxPhaseLength$ --- the length of a longest sequence of instructions that may compose a given phase, and;
\item $currentPhaseLength$ --- the number of time steps during which the processor has been executing the current phase, until the procedure's invocation.
\end{enumerate*}
Given these parameters, \algid{synch} adds a sequence of $maxPhaseLength - currentPhaseLength$ no-op instructions, guaranteeing that the number of steps taken from the beginning of each phase's execution until the end of its call is the same for all processors.
To use \algid{synch}, we rely on the purely theoretical procedure $\iexecd$ to obtain the value of $currentPhaseLength$.




For last, we partition a computation's execution even further, by partitioning all rounds \textbf{equally} into sequences of \emph{stages}.
To formalize this idea, define a \emph{stage partition} $\ddot{s} \in \mathbb{N} \times \mathbb{N}$, as $\ddot{s} = \left(base,\textit{offset}\right)$, with $\textit{offset} > 0$, where $base$ and $\textit{offset}$ are, respectively, the starting step and length of the stage defined by $\ddot{s}$ within each round.
Refer to the $i$-th stage of a round $\overline{t}$ as $\overline{t}\left\langle i \right\rangle$.

\begin{definition}
  \label{def:stages}
  Let $L$ be the length of the rounds.
  Say that a set $\ddot{S}$ is a set of \emph{stage partitions} if $L = \sum_{\ddot{s} \in \ddot{S}} \pi_{2}\parens[\small]{\ddot{s}}$ and $\forall \ddot{s} \in \ddot{S}\, \parens[\big]{\bracks{\exists \ddot{r} \in \ddot{S} :\,\,\, \pi_{1}\parens{\ddot{s}} + \pi_{2}\parens{\ddot{s}} = \pi_{1}\parens{\ddot{r}}} \vee \bracks{\pi_{1}\left(\ddot{s}\right) + \pi_{2}\left(\ddot{s}\right) = L}}$, where $\pi_{i}(t)$ denotes the projection of the $i$-th element of tuple $t$.
\end{definition}

\begin{remark}
  \label{remark:round-requirements-met-by-definition-of-scheduling-loop}
  To analyze a scheduler's performance using our methodology it suffices to:
  \begin{enumerate*}
  \item define the scheduler by a scheduling loop;
  \item divide the actions that processors take during each iteration of the loop (by partitioning each scheduling iteration into phases); and
  \item insert a call to the synch procedure at the end of each phase.
  \end{enumerate*}
\end{remark}
\begin{proof}[Justification]
    By using the \algid{synch} and $\iexecd$ procedures, one can guarantee that any scheduling algorithm, that may be defined by a scheduling loop, can be modified so that processors are kept synchronized throughout any computation's execution, having that all processors begin the execution of the $i$-th phase of the $n$-th scheduling iteration at the exact same step.
  With this, the length of each round can be set to $\sum_{i \in Phases} length_{i}$, where $Phases$ denotes the set of phases that compose a scheduling iteration and $length_{i}$ denotes the length of the $i$-th phase\footnote{Note that, by including the call to the \algid{synch} procedure at the end of each phase, we ensure that the $i$-th phase of every scheduling iteration has the same length.}.
  Note that, since all processors execute each scheduling iteration synchronized, the definition of scheduling loop ensures us that the requirements of the definition of round are satisfied:
  \begin{enumerate*}
  \item each round has constant length;
  \item a computation's execution can be partitioned into a sequence of equal-length rounds;
  \item during each round no processor executes more than a single node; and
  \item no node is migrated more than once during a round.
  \end{enumerate*}
  Then, it only remains to partition each round into a sequence of stages, having one stage per phase, and ensuring that the execution of the $i$-th phase of a scheduling iteration coincides with the $i$-th stage of the corresponding round.
\end{proof}

In the synchronous \ws\ scheduler, depicted in Algorithm~\ref{algo:ws}, \linebreak $max\_phase_{I}\_length$ and $max\_phase_{II}\_length$ are two constants that correspond to the lengths of the longest sequences of instructions composing the first and second phases of \ws, respectively.
Thus, by Remark~\ref{remark:round-requirements-met-by-definition-of-scheduling-loop}, we can set the length of \ws's rounds to $max\_phase_{I}\_length + max\_phase_{II}\_length$, and partition each such round into two stages whose length matches the maximum length of the corresponding phase.
To proceed to analysis of \ws's performance, it only remains to show that \ws\ satisfies Requirement~\ref{requirement:node-execution-constraints}.
For \ws, say that a node $\mu$ is attached to a processor $p$ if one of the following conditions holds:
\begin{enumerate*}
\item $\mu$ is $p$'s currently assigned node;
\item $\mu$ is stored in $p$'s deque; or
\item $\mu$ is stored in $enabled\left[0\right]$ or $enabled\left[1\right]$ (see line 4 of Algorithm~\ref{algo:ws}).
\end{enumerate*}
At the beginning of any round, each node that is attached to a processor is either in its deque or is the processor's currently assigned node.
As it can be observed in Algorithm~\ref{algo:ws}, each processor only executes the node that is stored in its \algid{assigned} variable.
Since the value of this variable is not changed at least until the processor executes the node, then the node was already stored in the \algid{assigned} variable when the round began, and so the requirement is satisfied.

\subsection{\wslong's performance}
\label{subsec:identifying-ws-limitation}
To show that the synchronous \ws\ algorithm (as defined in Algorithm~\ref{algo:ws}) is not \emph{Algorithm short-term stable}, we will create a computation for which work tends to accumulate unboundedly in some busy processors' deques.
Before moving to the actual proof, however, we have to make an additional definition.

\begin{definition}
\label{def:nodes-stolen}
Refer to the set of nodes stolen at step $i$ from a processor $p$ as $Stolen^{+}_{i}\left(p\right)$, and to the set of nodes stolen by $p$ as $Stolen^{-}_{i}\left(p\right)$.
Moreover, for some round $\overline{t}$, define the set of nodes stolen during $\overline{t}$ from $p$ as \[Stolen^{+}_{\overline{t}}\left(p\right) = \bigcup_{i \in \left\{\overline{t}\left[0\right],\ldots,\overline{t}\left[L-1\right]\right\}}Stolen^{+}_{i}\left(p\right)\] and the set of nodes stolen by $p$ as \[Stolen^{-}_{\overline{t}}\left(p\right) = \bigcup_{i \in \left\{\overline{t}\left[0\right],\ldots,\overline{t}\left[L-1\right]\right\}}Stolen^{-}_{i}\left(p\right)\]
\end{definition}

\begin{lemma}
  \label{lemma:ws-nodes-stolen-round}
  For \ws, at any round $\overline{t}$, $M^{+}_{\overline{t}}\left(p\right) = Stolen^{+}_{\overline{t}}\left(p\right)$ and $M^{-}_{\overline{t}}\left(p\right) = Stolen^{-}_{\overline{t}}\left(p\right)$.
\end{lemma}
\begin{proof}
  Both results follow from Definition~\ref{def:nodes-stolen} and the specification of Algorithm~\ref{algo:ws}.
\end{proof}

We now move to obtain both lower and upper bounds on the expected number of nodes that are stolen from a \textit{non-self-stable} processor during a round.

\begin{lemma}
  \label{lemma:bab-lower bounds on the probability that i is non-empty}
  Suppose there are $B$ bins and $B.\alpha$ balls, and that each ball is tossed independently and uniformly at random into the bins.
  For a bin $b_{i}$, let $Y_{i}$ be an indicator variable, defined as \[Y_{i} = \left\{\begin{matrix}
 \,1 & \text{if at least one ball lands in } b_{i};\\
 \,0 & \text{otherwise.}
\end{matrix}\right.\]
Then, $E\left[Y_{i}\right] = P\left\{Y_{i} = 1\right\} \geq 1 - \euler^{- \alpha}$.
\end{lemma}
\begin{proof}
  The probability that no ball lands in $b_{i}$ is $P\left\{Y_{i} = 0\right\} = \left(1 - \frac{1}{B}\right)^{B\alpha} \leq \euler^{- \alpha}$.
  To conclude, $E\left[Y_{i}\right] = P\left\{Y_{i} = 1\right\} \geq 1 - \euler^{- \alpha}$.
\end{proof}

\begin{lemma}
  \label{lemma:lower-and-upper-bounds-steals}
  For any round $\overline{t}$ and $p \in U_{\overline{t}}$ during a computation's execution using \ws, we have $1 - \euler^{- \alpha_{\overline{t}}} \leq \mathrm{E}\bracks{\abs{Stolen_{\overline{t}}^{+}\parens{p}}} \leq \alpha_{\overline{t}}$.
\end{lemma}
\begin{proof}
  By observing Algorithm~\ref{algo:ws}, it follows that a processor makes a steal attempt iff it is idle, implying that exactly $P\alpha_{\overline{t}}$ steal attempts are made during round $\overline{t}$.
  Note that:
  \begin{enumerate*}
  \item steal attempts are independent from one another; and
  \item a steal attempt corresponds to targeting a processor uniformly at random and then invoking the \algid{popTop} method to its deque.
  \end{enumerate*}
  If we imagine that each steal attempt is a ball toss and that each processor's deque is a bin, it follows by Lemma~\ref{lemma:bab-lower bounds on the probability that i is non-empty} that the probability of $p$'s deque being targeted is at least $1 - \euler^{- \alpha_{\overline{t}}}$.
  On the other hand, the expected number of invocations to the \algid{popTop} method of any processor $p$'s deque is $\left(P\alpha_{\overline{t}}\right)/P = \alpha_{\overline{t}}$.
  Since $p$ may only invoke the \algid{popBottom} method of its deque during the second phase and the all the steal attempts take place during the first phase, then, taking into account the deque semantics (see Section~\ref{sec:deque-semantics}):
  \begin{enumerate*}
  \item if $p$'s deque is targeted by at least one steal attempt, then at least one node is stolen; and
  \item at most one node might be returned for each invocation to the \algid{popTop} method.
  \end{enumerate*}
  Thus, $\mathrm{E}\bracks{\abs{Stolen_{\overline{t}}^{+}\parens{p}}} \geq 1 - \euler^{- \alpha_{\overline{t}}}$ and $\mathrm{E}\bracks{\abs{Stolen_{\overline{t}}^{+}\parens{p}}} \leq \alpha_{\overline{t}}$.
\end{proof}

\begin{lemma}
  \label{lemma:ws-lower-upper-bounds-migrations}
  For any round $\overline{t}$, and $p \in U_{\overline{t}}$ we have $1 - \euler^{- \alpha_{\overline{t}}} \leq \mathrm{E}\bracks{\abs{M_{\overline{t}}^{+}\parens{p}}} \leq \alpha_{\overline{t}}$, where  $\alpha_{\overline{t}}$ is the ratio of idle processors.
\end{lemma}
\begin{proof}
  Lemmas~\ref{lemma:ws-nodes-stolen-round} and~\ref{lemma:lower-and-upper-bounds-steals} imply this result. 
\end{proof}

The next lemma follows from the behavior of the \ws\ algorithm. 

\begin{lemma}
  \label{lemma:ws-base-results}
  Consider some processor $p \in Procs$ and some round $\overline{t}$ during the execution of a computation by \ws.
  If $p \in U_{\overline{t}}$ then $p$'s deque is non-empty and $M^{-}_{\overline{t}}\left(p\right) = \emptyset$.
  If $p \in S_{\overline{t}}$ then $\left|R_{\overline{t+1}}\left(p\right)\right| \leq 2$.
\end{lemma}
\begin{proof}
  By the definition of Algorithm~\ref{algo:ws} it can be proved by induction on the progression of a computation's execution that if a processor has at least one attached node at the beginning of round $\overline{t}$, then the processor executes a node during $\overline{t}$.
  From that, and by observing the algorithm, it follows that if $p$ has at least one node attached, then it does not make any steal attempt during $\overline{t}$, implying $Stolen^{-}_{\overline{t}}\left(p\right) = \emptyset$.
  Lemma~\ref{lemma:ws-nodes-stolen-round} then implies $M^{-}_{\overline{t}}\left(p\right) = \emptyset$.
  On the other hand, since $p$ always executes one of its attached nodes if there is any, it follows that if $p \in U_{\overline{t}}$ then $p$'s deque is not empty.

  If $p$ only has a single attached node, then $p \in S_{\overline{t}}$.
  Because it has one attached node, it follows $Stolen^{-}_{\overline{t}}\left(p\right) = \emptyset$.
  Again, Lemma~\ref{lemma:ws-nodes-stolen-round} then implies $M^{-}_{\overline{t}}\left(p\right) = \emptyset$.
  In addition, since the out-degree of any node is at most two (by our conventions regarding computations' structure), then at the end of the round $p$ has at most two attached nodes.
  
  Finally we show that if $R_{\overline{t}}\left(p\right) = \emptyset$, then at the end of the round $p$ has at most one attached node.
  If $p$ has no attached node, then its \algid{assigned} variable does not contain a valid node, implying that $p$ executes a call to the \algid{WorkMigration} procedure.
  Since each call only entails one invocation to the \algid{popTop} method, then, taking into account the method's semantics\footnote{More details on the appendix, Section~\ref{sec:deque-semantics}.} it follows that $p$ may only get at most one node from its steal attempt.
  Since after performing such attempt, $p$ takes no further action during the scheduling iteration other than simply waiting for it to end, we conclude the lemma holds.
\end{proof}

From Lemma~\ref{lemma:ws-base-results} and the definition of round (Definition~\ref{def:round}), it follows $\forall p \in S_{\overline{t}}, \left|R_{\overline{t+1}}\left(p\right)\right| \leq 2 \leq L + 1$, where $L$ is the length of the rounds, which is at least $1$.
Thus, if we were to show that \ws\ is \textit{algorithm short-term stable} wrt some interval $I \subseteq \left]0;1\right[$, then, by Corollary~\ref{corollary:connecting-corollary}, we would only have to prove that for any round $\overline{t}$ such that $\alpha_{\overline{t}} \in I$, we had $\left|E_{\overline{t}}\left(p\right)\right| < \mathrm{E}\bracks[\big]{\left|C_{\overline{t+1}}\left(p\right)\right| + \abs[\big]{M^{+}_{\overline{t}}\left(p\right)}}$.
Unfortunately, as we now prove, there is no non-empty interval $I$ wrt which \ws\ is \textit{algorithm short-term stable}.
%
%

\begin{theorem}
  \label{thr:ws-not-stable}
  There is no non-empty interval $I \subseteq \left]0;1\right[$ such that \ws\ (as defined in Algorithm~\ref{algo:ws}) is algorithm short-term stable wrt $I$.
\end{theorem}
\begin{proof}
Due to our conventions related with the computations' structure, it follows that during a round a processor can enable two nodes.
For some round $\overline{t}$, let $p$ be a non-self-stable processor (\textit{i.e.}~$p \in U_{\overline{t}}$) such that $\left|E_{\overline{t}}\left(p\right)\right| = 2$.
Lemmas~\ref{lemma:r-next} and \ref{lemma:ws-base-results} imply $R_{\overline{t+1}}\left(p\right) = \left(E_{\overline{t}}\left(p\right) \cup R_{\overline{t}}\left(p\right)\right) - (C_{\overline{t}}\left(p\right) \cup M^{+}_{\overline{t}}\left(p\right))$.
As already noted in the proof of Lemma~\ref{lemma:ws-base-results}, for the \ws\ algorithm, since $p \in U_{\overline{t}}$, it follows that $\left|C_{\overline{t}}\left(p\right)\right| = 1$.
Since we have
\begin{enumerate*}
\item $E_{\overline{t}}\left(p\right) \cup R_{\overline{t}}\left(p\right) \supseteq C_{\overline{t}}\left(p\right) \cup M^{+}_{\overline{t}}\left(p\right)$; \label{enum:claim-ws-part-1}
\item $E_{\overline{t}}\left(p\right) \cap R_{\overline{t}}\left(p\right) = \emptyset$; and  \label{enum:claim-ws-part-2}
\item \label{enum:claim-ws-part-3} $C_{\overline{t}}\left(p\right) \cap M^{+}_{\overline{t}}\left(p\right) = \emptyset$ \footnote{For a formal proof of these claims see Claims~\ref{claim:progression-connecting-lemma-4}, \ref{claim:progression-connecting-lemma-1} and \ref{claim:progression-connecting-lemma-3} for parts \ref{enum:claim-ws-part-1}, \ref{enum:claim-ws-part-2} and \ref{enum:claim-ws-part-3}, respectively, with $M^{-}_{\overline{t}}\left(p\right) = \emptyset$.},
\end{enumerate*}
then $\left|R_{\overline{t+1}}\left(p\right)\right| = \left|E_{\overline{t}}\left(p\right)\right| + \left|R_{\overline{t}}\left(p\right)\right| - \left|C_{\overline{t}}\left(p\right) \right| - |M^{+}_{\overline{t}}\left(p\right)|$.
By Lemma~\ref{lemma:ws-lower-upper-bounds-migrations} and since $p$ enabled two nodes, it follows $\mathrm{E}\bracks{\left|R_{\overline{t+1}}\left(p\right)\right|} \geq 2 + \left|R_{\overline{t}}\left(p\right)\right| - 1 - \alpha_{\overline{t}}= \left|R_{\overline{t}}\left(p\right)\right| + 1 - \alpha_{\overline{t}}$.
Since $p$ enabled two nodes, it executes a node during the next round, implying $\left|C_{\overline{t+1}}\left(p\right)\right| = 1$.
It then follows, $\mathrm{E}\bracks{\left|R_{\overline{t+1}}\left(p\right) - C_{\overline{t+1}}\left(p\right)\right|} \geq \left|R_{\overline{t}}\left(p\right)\right| + 1 - \alpha_{\overline{t}} - 1 =\left|R_{\overline{t}}\left(p\right)\right| - \alpha_{\overline{t}}$.
Even though the definition of \textit{algorithm short-term stability} only considers ratios of idle processors in $]0;1[$, note that for \ws, if all processors are idle, then the computation's execution must have already finished, and so it only makes sense to analyze rounds during which the execution is still ongoing.
It then follows that $\alpha_{\overline{t}} < 1$, implying $\mathrm{E}\bracks{\abs{R_{\overline{t+1}}\parens{p} - C_{\overline{t+1}}\parens{p}}} \geq \left|R_{\overline{t}}\left(p\right)\right| - \alpha_{\overline{t}} > \abs{R_{\overline{t}}\parens{p} - C_{\overline{t}}\parens{p}}$.
Thus, if during round $\overline{t}$, $p$ were the only non-self-stable processor (\textit{i.e.}~$U_{\overline{t}} = \{p\}$), then $\mathrm{E}\bracks{\abs{R_{\overline{t+1}}\parens{U_{\overline{t}}} - C_{\overline{t+1}}\parens{U_{\overline{t}}}}} > \abs{R_{\overline{t}}\parens{U_{\overline{t}}} - C_{\overline{t}}\parens{U_{\overline{t}}}}$.
\end{proof}

As one might note, this implies that \ws\ cannot even guarantee \textit{short-term stability} for the set $\{p\}$ (recall Definition~\ref{def:short-term-stability}), regardless of the ratio of idle processors.



\section{A greedy \wsslong\ algorithm}
\label{sec:analyzing-wss}

In this section we present and analyze the performance of a (purely theoretical) greedy \wsslong\ scheduler --- or simply \wss.
This algorithm  (depicted in Algorithm~\ref{algo:wss}) is a variant of \ws\ where processors load balance not only by stealing work, but also by spreading it.
As in \ws, each processor owns a lock-free deque (obeying the semantics defined in~\cite{DBLP:journals/mst/AroraBP01}) and a variable \algid{assigned} that stores the node that it will execute next, if any.
To implement the spreading mechanism each processor additionally owns a \algid{state} flag and a \algid{donation} cell.
Processors use the \algid{state} flag to inform other processors on their current state --- \algid{working}, \algid{idle} or marked as target of a donation (more on this ahead) --- and use the \algid{donation} cell to store nodes that they want to spread.
In \wss\ processors are uniquely identified by an \algid{id}, with which they can be accessed in constant time.
The scheduler also makes use of the \algid{CAS} instruction (Compare-And-Swap), with its usual semantics. 
Thus, at most one \algid{CAS} instruction targeting the same memory location can successfully execute at each step.
We assume that the processor that succeeds executing the \algid{CAS} instruction over a memory address \algid{m} at some step \algid{i} is chosen uniformly at random from the set of processors that are eligible to successfully execute the instruction at step \algid{i} over memory address \algid{m}.

Contrarily to \ws, we partition each scheduling iteration of \wss\ into three phases.
Phase I of \wss\ is very similar to the \ws's counterpart, only differing because in \wss\ processors keep updating their \algid{state} flags to reflect their current state.
Phases II and III of \wss\ are as follows:
\begin{description}[leftmargin=0cm,noitemsep,topsep=5pt,parsep=5pt,partopsep=0pt]
\item[Phase II]
  If, in phase I, a processor $p$ made a steal attempt or executed a node that did not enable any node, then $p$ does not take any action during this phase.
  Otherwise, if at least one node was enabled, one of the enabled nodes becomes $p$'s new assigned node.
  If two nodes were enabled, then, after having a new node assigned, $p$ attempts to spread the node it did not assign.
\item[Phase III]
  If a processor $p$ executed a node in phase I but no node was enabled, $p$ invokes \algid{popBottom} to fetch the bottommost node from its deque, if there is any.
  On the other hand, if a single node was enabled, $p$ does not take any action during this phase.
  If two nodes were enabled, $p$ only takes action if the donation attempt it made during phase II failed.
  In such scenario, $p$ pushes the node it failed to donate into the bottom of its own deque, via the \algid{pushBottom} method.
  Finally, if the processor made an unsuccessful steal attempt during the first phase, it polls its \algid{state} flag to check for incoming donations.
  If there is a donation, $p$ transfers the node from the donor's \algid{donation} cell and updates its \algid{state} flag accordingly.  
\end{description}


  \begin{algorithm}[t]
    \caption{The \wss\ algorithm}
    \label{algo:wss}
    \begin{small}
      \begin{algorithmic}[1]
        \Procedure{Scheduler}{ }
        \While{\textbf{not} finished ($computation$)}
        \If{ValidNode ($self.assigned$)}
        \State\ $enabled \gets$ execute ($self.assigned$)\;
        \State\ $assigned \gets $ \textsc{none}\;
        \State\ synch($max\_phase_{I}\_length$, $\iexecd()$)\;
        \If{length ($enabled$) $> 0$}
        \State\ $self.assigned \gets enabled\left[0\right]$\;
        \If{length ($enabled$) = $2$}
        \State\ $self$.handleExtraNode ($enabled\left[1\right]$)\;
        \Else
        \State\ synch($max\_phase_{II}\_length$, $\iexecd()$)\;
        \EndIf\
        
        \Else\
        \State\ synch($max\_phase_{II}\_length$, $\iexecd()$)\;
        \State\ $self.assigned \gets self.deque$.popBottom ()\;
        \If{\textbf{not} ValidNode ($self.assigned$)}
        \State\ $self.state \gets $ \textsc{idle}\;
        \EndIf\
        \EndIf\
        \Else
        \State\ $self$.loadBalance ()\;
        \EndIf\
        \State\ synch($max\_phase_{III}\_length$, $\iexecd()$)\;
        \EndWhile\
        \EndProcedure\
         \newline
        \Procedure{handleExtraNode}{$\mu$}
        \State\ $self.donation \gets \mu$\;
        \State\ $donee \gets$ UniformlyRandomProcessor()\;
        \State\ $result \gets$ CAS ($donee.state$, \textsc{idle}, $self.id$)\;
        \State\ synch($max\_phase_{II}\_length$, $\iexecd()$)\;
        
        \If{$result \neq $ \textsc{success}}
        \State\ $self.donation \gets$ \textsc{none}\;
        \State\ $self.deque$.pushBottom ($\mu$)\;
        \EndIf\
        \EndProcedure\
        \newline
        \Procedure{loadBalance}{ }
        \State\ $victim \gets$ UniformlyRandomProcessor()\;
        \State\ $self.assigned \gets victim.deque$.popTop ()\;
        \If{ValidNode ($self.assigned$)}
        \State\ $self.state \gets $ \textsc{working}\;
        \EndIf\
        
        \State\ synch($max\_phase_{I}\_length$, $\iexecd()$)\;
        \State\ synch($max\_phase_{II}\_length$, $\iexecd()$)\;
        
        \If{$self.state \neq $ \textsc{idle} and $self.state \neq$ \textsc{working}}
        \State\ $donor \gets processor\left[self.state\right]$\;
        \State\ $self.assigned \gets donor.donation$\;
        \State\ $self.state \gets $ \textsc{working}\;
        \EndIf\
        \EndProcedure\
      \end{algorithmic}
    \end{small}
  \end{algorithm}

\begin{definition}
\label{def:nodes-spread}
Refer to the set of nodes spread at step $i$ by a processor $p$ as $Spread^{+}_{i}\left(p\right)$, and to the set of nodes spread to $p$ as $Spread^{-}_{i}\left(p\right)$.
Moreover, for some round $\overline{t}$, define the set of nodes spread during $\overline{t}$ by $p$ as \[Spread^{+}_{\overline{t}}\left(p\right) = \bigcup_{i \in \left\{\overline{t}\left[0\right],\ldots,\overline{t}\left[L-1\right]\right\}}Spread^{+}_{i}\left(p\right),\] and the set of nodes spread to $p$ as \[Spread^{-}_{\overline{t}}\left(p\right) = \bigcup_{i \in \left\{\overline{t}\left[0\right],\ldots,\overline{t}\left[L-1\right]\right\}}Spread^{-}_{i}\left(p\right).\]
\end{definition}

The next claim implies that we can use the methodology described in Section~\ref{sec:strategy-randomized-analysis} to analyze the performance of \wss.

\begin{claim}
  \label{claim:wss-meets-requirements}
  The \wss\ algorithm can be defined using a scheduling loop and meets Requirement~\ref{requirement:node-execution-constraints}.
\end{claim}
\begin{proof}[Proof Sketch]
As one can observe from Algorithm~\ref{algo:wss}, like \ws, \wss\ can also be naturally defined using scheduling loops (lines 2 to 24):
\begin{enumerate*}
\item at most one node is executed per scheduling iteration;
\item if a node is migrated to a processor during an iteration, then it is not migrated again;
\item the length of every iteration is bounded by a constant; and
\item the full sequence of instructions executed by any processor can be partitioned into a sequence of scheduling iterations.
\end{enumerate*}
As for \ws, and for the same reasons, we do not formally show that the loop starting at line 2 and ending at line 25 of Algorithm~\ref{algo:wss} satisfies the requirements of a scheduling loop.
Nevertheless, it is easy to deduce, by observing the definition of the scheduler (given in Algorithm~\ref{algo:wss}), that this claim holds.
\end{proof}

Taking into account Claim~\ref{claim:wss-meets-requirements} and Remark~\ref{remark:round-requirements-met-by-definition-of-scheduling-loop}, we can begin \wss's analysis.
First, we show that for \wss, a node is migrated iff it is stolen or spread.

\begin{lemma}
  \label{lemma:wss-nodes-stolen-spread-round}
  For \wss, at any round $\overline{t}$ and for any processor $p$, $M^{+}_{\overline{t}}\left(p\right) = Stolen^{+}_{\overline{t}}\left(p\right) \cup Spread^{+}_{\overline{t}}\left(p\right)$ and $M^{-}_{\overline{t}}\left(p\right) = Stolen^{-}_{\overline{t}}\left(p\right) \cup Spread^{-}_{\overline{t}}\left(p\right)$.
\end{lemma}
\begin{proof}
  Both results follow from Definition~\ref{def:nodes-stolen} and Algorithm~\ref{algo:wss}.
\end{proof}

The next lemma is analogous to Lemma~\ref{lemma:ws-base-results}, but concerning the \wss\ algorithm.

\begin{lemma}
  \label{lemma:wss-base-results}
  Consider some $p \in Procs$ and some round $\overline{t}$ during the execution of a computation by \wss.
  Then:
  \begin{enumerate*}
  \item if $p \in U_{\overline{t}}$ then $p$'s deque is non-empty and $M^{-}_{\overline{t}}\left(p\right) = \emptyset$; and
  \item if $p \in S_{\overline{t}}$ then $\left|R_{\overline{t+1}}\left(p\right)\right| \leq 2$.
  \end{enumerate*}
\end{lemma}
\begin{proof}
  This proof follows the same general arguments as the proof of Lemma~\ref{lemma:ws-base-results}.
  
  As for \ws, taking into account the definition of the \wss\ scheduler (depicted in Algorithm~\ref{algo:wss}) it can be proved by induction on the progression of the computation's execution that if a processor has at least one attached node at the beginning of round $\overline{t}$, then the processor executes a node during $\overline{t}$.
  From that, and by observing the algorithm, it follows that if $p$ has at least one node attached, then:
  \begin{enumerate*}
  \item $p$ does not make any steal attempt during $\overline{t}$, implying $Stolen^{-}_{\overline{t}}\left(p\right) = \emptyset$; and
  \item $p$'s \algid{state} flag is set to \textsc{working} at least until the beginning of the third stage of round $\overline{t}$, implying that no processor donates work to $p$, and so $Spread^{-}_{\overline{t}}\left(p\right) = \emptyset$.
  \end{enumerate*}
  Thus, taking into account Lemma~\ref{lemma:wss-nodes-stolen-spread-round}, if $p$ has at least one attached node then $M^{-}_{\overline{t}}\left(p\right) = \emptyset$.
  
  To conclude the proof of the first statement of this lemma, note that because $p$ always executes one of its attached nodes as long as there is any, it follows that if $p \in U_{\overline{t}}$ then $p$ has at least two nodes attached and so $p$'s deque can not be empty.
  
  Again, since $p$ executes one of its attached nodes as long as there is any, if $p$ only has a single attached node then $p \in S_{\overline{t}}$ and $M^{-}_{\overline{t}}\left(p\right) = \emptyset$.
  By the nodes' out-degree assumption, it then follows that at the end of round $\overline{t}$, $p$ can have at most two attached nodes.
  
  Finally we show that if $R_{\overline{t}}\left(p\right) = \emptyset$, then at the end of the round $p$ has at most two attached nodes.
  If $p$ has no attached node, then its \algid{assigned} variable does not contain a valid node, implying that $p$ executes a call to the \algid{LoadBalance} procedure (line 22).
  Since each call only entails one invocation to the \algid{popTop} method, then, taking into account the method's semantics (see Section~\ref{sec:deque-semantics}) it follows that $p$ may only get at most one node from its steal attempt.
  On the other hand, as one can deduce from the definition of the \algid{LoadBalance} procedure, $p$ can only accept at most one node donation during the procedure's invocation\footnote{In fact, $p$ only accepts a donation if its steal attempt failed, and so $p$ can only have at most one attached node.}.
  Thus, at the end of the call $p$ can only have at most two attached nodes.
  To conclude the proof of the second part of the lemma, note that, after the call to the \algid{LoadBalance} procedure returns, $p$ takes no further action during the iteration.
\end{proof}

With this, we start obtaining bounds on the expected number of nodes that are migrated during a round for \wss.
To begin, we obtain both lower and upper bounds on the expected number of nodes stolen for \wss.

\begin{lemma}
  \label{lemma:lower-and-upper-bounds-steals-wss}
  For any round $\overline{t}$ and $p \in U_{\overline{t}}$ during a computation's execution using \wss, $1 - \euler^{- \alpha_{\overline{t}}} \leq \mathrm{E}\bracks{\abs{Stolen_{\overline{t}}^{+}\parens{p}}} \leq \alpha_{\overline{t}}$.
\end{lemma}
\begin{proof}
  The proof of this result is identical to the proof of Lemma~\ref{lemma:lower-and-upper-bounds-steals}, following from Lemma~\ref{lemma:wss-base-results} and the definition of \wss\ (see Algorithm~\ref{algo:wss}).
\end{proof}

Next, we obtain lower bounds on the expected number of nodes spread by any processor that enables two nodes.
A full proof of this result can be found in the Appendix (Section~\ref{sub:5-proof-lemma-1}).

\begin{lemma}
  \label{lemma:lower-bounds-expected-spreads-p}
  $\forall p \in U_{\overline{t}}$, if $\left|E_{\overline{t}}\left(p\right)\right| = 2$ then $\mathrm{E}\bracks{|Spread^{+}_{\overline{t}}\left(p\right)|} \geq \frac{\alpha_{\overline{t}}^{2}}{1 - \alpha_{\overline{t}}} \left(1 - \euler^{- \left(1 - \alpha_{\overline{t}}\right)}\right)$.
\end{lemma}

The following Lemma, together with Lemmas~\ref{lemma:lower-and-upper-bounds-steals-wss} and \ref{lemma:lower-bounds-expected-spreads-p}, allow us to obtain bounds on the expected number of nodes that are migrated from a processor.
A full proof can be found in the Appendix (Section~\ref{sub:5-proof-lemma-2}).
\begin{lemma}
  \label{lemma:wss-nodes-stolen-spread-round-intersection}
  For \wss, at any round $\overline{t}$, $Stolen^{+}_{\overline{t}}\left(p\right) \cap Spread^{+}_{\overline{t}}\left(p\right) = \emptyset$.
\end{lemma}

The next result states an inequality that will be used to prove that \wss\ is \emph{algorithm short-term stable} wrt interval $\left[0,7375; 1\right[$, and its full proof can be found in the Appendix (Section~\ref{sub:5-proof-lemma-3}).

\begin{lemma}
\label{lemma:wss-inequation-proof}
$\forall \alpha \in \left[0,7375; 1\right[, \quad 1 < 1 - \euler^{- \alpha} + \frac{\alpha^{2}}{1 - \alpha} \left(1 - \euler^{- (1 - \alpha)}\right)$.
\end{lemma}

Finally, we can prove that \wss\ overcomes the limitations of \ws.

\begin{theorem}
  \label{thr:wss-algorithm-short-term-stability}
  \wss\ (as defined in Algorithm~\ref{algo:wss}) is algorithm short-term stable wrt $[0,7375;1[$.
\end{theorem}
\begin{proof}
  Recall that $L$ denotes the length of each round, and that by definition $L \geq 1$.
  Then, from Lemma~\ref{lemma:wss-base-results}, it follows that $\forall p \in S_{\overline{t}}, \left|R_{\overline{t+1}}\left(p\right)\right| \leq 2 \leq L + 1$.
  Furthermore, taking into account Lemma~\ref{lemma:stability-generalize} and Corollary~\ref{corollary:connecting-corollary}, it follows that to prove this theorem it suffices to show that for any round $\overline{t}$ such that $\alpha_{\overline{t}} \in [0,7375;1[$, we have $\forall p \in U_{\overline{t}}, \quad \left|E_{\overline{t}}\left(p\right)\right| < \mathrm{E}\bracks{\left|C_{\overline{t+1}}\left(p\right)\right| + \abs[\big]{M^{+}_{\overline{t}}\left(p\right)}}$.

  For an arbitrary round $\overline{t}$, consider any processor $p \in U_{\overline{t}}$ (\textit{i.e.} $p$ is a processor that is \emph{non-self-stable} at round $\overline{t}$).
  Due to our conventions related with computations' structure, $\left|E_{\overline{t}}\left(p\right)\right|$ is either $0$, $1$, or $2$:
  \begin{itemize}[leftmargin=1cm,noitemsep,topsep=5pt,parsep=5pt,partopsep=0pt]
  \item If $\left|E_{\overline{t}}\left(p\right)\right| = 0$ then, by Lemma~\ref{lemma:lower-and-upper-bounds-steals-wss} it follows $\left|E_{\overline{t}}\left(p\right)\right| < \mathrm{E}\bracks{\left|C_{\overline{t+1}}\left(p\right)\right| + \abs[\big]{M^{+}_{\overline{t}}\left(p\right)}}$.
  \item If $\left|E_{\overline{t}}\left(p\right)\right| = 1$ then, by the specification of the \wss\ algorithm, it follows $\left|C_{\overline{t+1}}\left(p\right)\right| = 1$.
    Taking into account Lemma~\ref{lemma:lower-and-upper-bounds-steals-wss}, we deduce $\left|E_{\overline{t}}\left(p\right)\right| < \mathrm{E}\bracks{\left|C_{\overline{t+1}}\left(p\right)\right| + \abs[\big]{M^{+}_{\overline{t}}\left(p\right)}}$.
  \item By the specification of the \wss\ algorithm, it follows that if $\left|E_{\overline{t}}\left(p\right)\right| = 2$ then $\left|C_{\overline{t+1}}\left(p\right)\right| = 1$.
    Thus, to prove this case it suffices to show $1 < \mathrm{E}\bracks{\abs[\big]{M^{+}_{\overline{t}}\left(p\right)}}$.
    We now prove just that.
    By Lemmas~\ref{lemma:wss-nodes-stolen-spread-round} and \ref{lemma:wss-nodes-stolen-spread-round-intersection}, it follows $\abs[\big]{M^{+}_{\overline{t}}\left(p\right)} = \abs[\big]{Stolen^{+}_{\overline{t}}\left(p\right)} + \abs[\big]{Spread^{+}_{\overline{t}}\left(p\right)}$,
and by Lemmas~\ref{lemma:lower-and-upper-bounds-steals-wss} and \ref{lemma:lower-bounds-expected-spreads-p}, it follows
$\mathrm{E}\bracks{\abs[\big]{M^{+}_{\overline{t}}\left(p\right)}} \geq 1 - \euler^{- \alpha_{\overline{t}}} + \frac{\alpha_{\overline{t}}^{2}}{1 - \alpha_{\overline{t}}} \left(1 - \euler^{- \left(1 - \alpha_{\overline{t}}\right)}\right)$.
Thus, taking into account Lemma~\ref{lemma:wss-inequation-proof}, having $\alpha = \alpha_{\overline{t}}$, we conclude the proof of Theorem~\ref{thr:wss-algorithm-short-term-stability}.
  \end{itemize}
  

\end{proof}




\section{Related work}
\label{sec:related-work}
To the best of our knowledge, there is no work that analyzes the performance of online structured computation schedulers, on a round basis, depending solely on the ratio of idle processors.

Most theoretical work dealing with the study of online structured computation schedulers, has focused on proving properties related with the (complete) execution of computations by \ws\ and variants.
Blumofe~\textit{et al.} proved that \ws\ is optimal up to a constant factor in terms of space requirements, expected execution time, and expected communication costs~\cite{DBLP:journals/jacm/BlumofeL99}.
Arora~\textit{et al.} showed that \ws\ is optimal even for multiprogrammed environments~\cite{DBLP:conf/spaa/AroraBP98,DBLP:journals/mst/AroraBP01}. 
Agrawal~\textit{et al.} introduced a variant of \ws\ that avoids unnecessary load balancing cycles in order to achieve higher efficiency~\cite{DBLP:conf/ipps/AgrawalHHL07,DBLP:journals/tocs/AgrawalLHH08}.
The authors proved that \ws\ is capable of maintaining nearly optimal bounds, while reducing the number of cycles during which processors are not making progress on a computation's execution (corresponding to load balancing cycles), down to a constant factor away from the computation's total amount of work.
Regarding data locality, Acar~\textit{et al.} obtained both lower and upper bounds on the number of cache misses using \ws~\cite{DBLP:journals/mst/AcarBB02}.
More recent research has been focusing on reducing the synchronization overheads of \ws~\cite{DBLP:conf/ppopp/AcarCR13}, mainly by eliminating synchronization for local deque operations (\textit{i.e.}~eliminating the need for synchronization when processors work locally on their own deque).
Even more recently, Muller~\textit{et al.} studied the performance of \ws\ for computations that include latency operations (such as receiving input from a user), obtaining promising results~\cite{DBLP:conf/spaa/MullerA16}.
On the other hand, most practical work that deals with the scheduling of structured computations has focused either on the improvement of current \ws\ implementations --- increasing data locality~\cite{DBLP:journals/mst/AcarBB02,DBLP:conf/ipps/GuoZCS10,DBLP:conf/europar/QuintinW10,DBLP:journals/ipl/SuksompongLS16}, reducing synchronization overheads~\cite{DBLP:conf/ppopp/AcarCR13,DBLP:conf/ppopp/HiraishiYUY09,DBLP:conf/ppopp/MichaelVS09,DBLP:conf/asplos/MorrisonA14,DBLP:conf/europar/DijkP14}, etc ---, or on the development of libraries and languages implementing \ws\ on both shared memory environments~\cite{pasl:parallel-algorithm-scheduling-library,DBLP:journals/jpdc/BlumofeJKLRZ96,DBLP:journals/sigarch/Faxen08,DBLP:conf/dac/Leiserson09,DBLP:conf/spaa/MullerA16} and distributed settings~\cite{DBLP:conf/icpp/CongKKLSW08,DBLP:conf/sc/DinanLSKN09,DBLP:conf/hpdc/LifflanderKK12,DBLP:conf/europar/QuintinW10}.

While, for the execution of structured computations, work generation depends on what has already been executed, for independent task scheduling, work generation (or, more correctly, task arrival) is assumed to be independent from what tasks processors already executed~\cite{DBLP:journals/rsa/AdlerCMR98,DBLP:journals/siamcomp/AzarBKU99,DBLP:journals/siamcomp/BerenbrinkFG03,DBLP:conf/spaa/LulingM93,DBLP:conf/spaa/Mitzenmacher98,DBLP:conf/focs/MitzenmacherPS02}. 
In fact, much of the work in this area consists on studying the effectiveness of different strategies (that rely on randomness) for placing $n$ balls (each representing a task) into $n$ bins (each representing a processor)~\cite{DBLP:journals/rsa/AdlerCMR98,DBLP:journals/siamcomp/AzarBKU99,DBLP:conf/spaa/LulingM93,DBLP:conf/focs/MitzenmacherPS02}, being that a strategy's effectiveness is measured according to the number of balls that the fullest bin is expected to have: the lower this number is, the more effective the strategy is.
Of course, this type of models, despite being suitable for modelling independent task schedulers, are far from being apt to model the performance of structured computation schedulers (for example, note that in the execution of a structured computation, work is generated per processor).
Within the area of independent task scheduling, perhaps the work most closely related to ours is on the performance analysis of online independent task schedulers~\cite{DBLP:journals/siamcomp/BerenbrinkFG03,DBLP:conf/spaa/Mitzenmacher98}.
Yet, to the best of our knowledge, all the analyzes made to these schedulers rely upon the assumption that tasks arrive to the system according to some random distribution (typically Poisson's distribution).
For instance, Mitzenmacher proposed a simple but powerful scheme to analyze \textbf{independent task} work stealing schedulers, that uses differential equations~\cite{DBLP:conf/spaa/Mitzenmacher98}.
This scheme allows to study not only the most basic work stealing schedulers (of independent tasks), but also more complex variants (\textit{e.g.}~allowing processors to repeat a steal whenever its steal attempt aborted).
Nevertheless, the proposed scheme relies on the assumption that work is generated according to some random distribution, and so it is not suitable for modelling the behavior of structured computation schedulers.
Berenbrink~\textit{et al.} study the performance of independent task work stealing schedulers, modelling the system as a Markov chain, whose states denote the number of tasks attached to each processor of the system~\cite{DBLP:journals/siamcomp/BerenbrinkFG03}.
The authors proved that the work stealing scheduler for independent tasks, where each steal is allowed to take up to half of a processor's work, is stable for a long term execution.
Unfortunately, their analysis also relies on the assumption that tasks arrive at the system according to a random distribution, and so it is not apt to model the performance of structured computation schedulers.
In addition, the authors assume that the number of tasks generated at each round is at most the number of processors, which, taking into account the standard conventions regarding the structure of computations~\cite{DBLP:conf/spaa/AroraBP98,DBLP:journals/mst/AroraBP01,DBLP:journals/jacm/BlumofeL99,DBLP:conf/ppopp/AcarCR13}, is not realistic for modelling schedulers of structured computations.

Although it may not be entirely straightforward, it is possible to use our methodology to model the steal-half work stealing algorithm.
To do so, each steal would have to be divided into a sequence of scheduling iterations, such that during each iteration the thief transferred a node from its victim.
However, transferring half of a processor's work may take some time, which not only implies that the thief will have to wait until it can begin executing what it stole, but it also means that either concurrent steal attempts to the same deque are delayed (to avoid duplicate steals), or thieves have to first transfer the work they intend to steal from their victims and only then attempt to commit the steal.
Regarding the later option, note that if a thief is transferring work from one of the only processors that is generating work, then the steal attempt is likely to fail.
Moreover, since during each round a processor can enable two nodes, then, it would still be possible that the processor whose deque was being stolen generated a large amount of work.

\section{Conclusion}
\label{sec:conclusion}
We introduced a formal framework for the performance analysis of structured computation schedulers, and defined an appropriate criterion for measuring the performance of online scheduling algorithms: \textit{algorithm short-term stability}.
Moreover, we introduced a simple and powerful method that allows to analyze the performance of these schedulers, and have demonstrated its convenience by using it with two different ends:
\begin{enumerate*}
\item proving that the performance of \ws\ is indeed limited; and
\item analyzing the performance of \wss.
\end{enumerate*}
Although \wss\ is a purely theoretical algorithm, its analysis gave us insight on how to possibly overcome the limitation of \ws.
Nevertheless, the greedy spreading strategy of the algorithm has a severe limitation that makes us question its practical value:
even if every processor is busy, whenever a processor generates work it makes a spread attempt.
This not only makes processors incur in unnecessary overheads (that, for modern computer architectures, are unduly large) but even more importantly, it entails a serious drawback concerning the communication costs of the algorithm.
Consequently, it is still an open problem to come up with a practical algorithm that overcomes \ws's limitation while maintaining its asymptotically optimal expected execution time and communication costs, and its low space requirements.

\bibliographystyle{alpha}
\bibliography{bibliography}

\clearpage

\appendix

\section{Full proofs for the results obtained in Section~1}
\label{sec:proofs-section-preliminaries}

\subsection{Full proof for Lemma~\ref{lemma:migrated-to-p-subset-migrated-from-all-but-p}}
\label{sub:1-proof-lemma-1}
\begin{claim}
  \label{claim:migrated-to-p-subset-migrated-from-all-but-p-step}
  For any step $i$ and processor $p$, $M^{-}_{i}\parens{p} \subseteq M^{+}_{i}\parens{Procs - \braces{p}}$.
\end{claim}
\begin{proof}
  By Definition~\ref{def:node-definitions}
  \begin{align*}
    M^{-}_{i}\parens{p} &= R_{i+1}\parens{p} \cap \parens{R_{i} - R_{i}\parens{p}} \\
                        &= \left(R_{i+1} - \left[\bigcup_{q \in Procs - \left\{p\right\}}R_{i}\left(q\right)\right]\right) \cap \left[\bigcup_{q \in Procs - \left\{p\right\}}R_{i}\left(q\right)\right] \\
                        &\subseteq R_{i+1} \cap \left[\bigcup_{q \in Procs - \left\{p\right\}}R_{i}\left(q\right)\right] \\
                        &\subseteq \left[\bigcup_{q \in Procs - \left\{p\right\}}R_{i}\left(q\right) \cap R_{i+1} \cap \overline{R_{i+1}\left(q\right)}\right] \\
                        &= M^{+}_{i}\left(Procs - \left\{p\right\}\right)
  \end{align*}
\end{proof}

\begin{proof}[Proof of Lemma~\ref{lemma:migrated-to-p-subset-migrated-from-all-but-p}]
  Claim~\ref{claim:migrated-to-p-subset-migrated-from-all-but-p-step} implies that for any step $i \in \braces{\overline{t}\bracks{0},\cdots,\overline{t}\bracks{L - 1}}$, we have $M^{-}_{i}\parens{p} \subseteq M^{+}_{i}\parens{Procs - \braces{p}}$.
  Thus, by Definition~\ref{def:node-definitions-round} we conclude this lemma holds.
\end{proof}

\subsection{Full proof for Lemma~\ref{lemma:r-next} (Round Progression Lemma)}
\label{sub:1-proof-lemma-2}
\begin{claim}
  \label{claim:r-next-requirement-claim-1}
  For any round $\overline{t}$ and processor $p \in Procs$, $R_{\overline{t+1}}\left(p\right) \cap M^{+}_{\overline{t}}\left(p\right) = \emptyset$.
\end{claim}
\begin{proof}
  For the purpose of contradiction, assume $R_{\overline{t+1}}\left(p\right) \cap M^{+}_{\overline{t}}\left(p\right) \neq \emptyset$.
  Thus, there is a step $j \in \left\{\overline{t}\left[0\right],\ldots,\overline{t}\left[L-1\right]\right\}$ such that $R_{\overline{t+1}}\left(p\right) \cap M^{+}_{j}\left(p\right) \neq \emptyset$.
  For such step $j$, let $S = R_{\overline{t+1}}\left(p\right) \cap M^{+}_{j}\left(p\right)$.
  Then,
  \begin{align*}
    S &= R_{\overline{t+1}\left[0\right]}\left(p\right) \cap M^{+}_{j}\left(p\right) \\
      &= R_{\overline{t+1}\left[0\right]}\left(p\right) \cap \left(R_{j}\left(p\right) \cap \left(R_{j+1} - R_{j+1}\left(p\right)\right)\right) \\
      &= R_{\overline{t+1}\left[0\right]}\left(p\right) \cap R_{j}\left(p\right) \cap R_{j+1} \cap \overline{R_{j+1}\left(p\right)}
  \end{align*}
  If $j$ were $\overline{t}\left[L-1\right]$, then $S = \emptyset$, and so, as one can deduce, $j < \overline{t}\left[L-1\right]$.
  Now, consider a node $\mu \in S$.
  It follows, $\mu \in R_{\overline{t+1}\left[0\right]}\left(p\right) \cap R_{j}\left(p\right) \cap R_{j+1} \cap \overline{R_{j+1}\left(p\right)}$.
  Since a node that is \textbf{ready} can only become \textbf{executed}, and a node in state \textbf{executed} does not change its state, it follows $\forall i \in \left\{j,\ldots,\overline{t+1}\left[0\right]\right\}, \mu \in R_{i}$.
  Moreover, as $\mu \in \overline{R_{j+1}\left(p\right)} \cap R_{s_{l}}\left(p\right)$ and $j + 1 < \overline{t}\left[L-1\right]$, it follows that there is a step $k \in \left\{j+1,\ldots,\overline{t}\left[L-1\right]\right\}$ such that $\mu \in R_{k+1}\left(p\right) \cap \overline{R_{k}\left(p\right)} \cap R_{k}$.
  By Definition~\ref{def:node-definitions}, it follows $\mu \in M^{-}_{k}\left(p\right)$, implying $\mu \in M^{-}_{\overline{t}}\left(p\right)$.
  However, since $\mu \in M^{-}_{\overline{t}}\left(p\right)$ and $\mu \in M^{+}_{\overline{t}}\left(p\right)$, it follows $M^{-}_{\overline{t}}\left(p\right) \cap M^{+}_{\overline{t}}\left(p\right) \neq \emptyset$, which, together with Lemma~\ref{lemma:migrated-to-p-subset-migrated-from-all-but-p}, contradicts Definition~\ref{def:round} --- the definition of rounds.
\end{proof}

\begin{lemma}
\label{lemma:step-progression-lemma}
For any step $i$ and $p \in Procs$, $R_{i+1}\left(p\right) \subseteq R_{i}\left(p\right) \cup E_{i}\left(p\right) \cup M^{-}_{i}\left(p\right)$.
\end{lemma}
\begin{proof}[Proof of Lemma~\ref{lemma:step-progression-lemma}]
  By Definition~\ref{def:node-definitions}, it follows
  \begin{align*}
    &R_{i}\left(p\right) \cup E_{i}\left(p\right) \cup M^{-}_{i}\left(p\right) \\
    &= R_{i}\left(p\right) \cup \left(R_{i+1}\left(p\right) - R_{i}\right) \cup \left(R_{i+1}\left(p\right) \cap \left(R_{i} - R_{i}\left(p\right)\right)\right) \\
    &= R_{i}\left(p\right) \cup \left(R_{i+1}\left(p\right) \cap \overline{R_{i}}\right) \cup \left(R_{i+1}\left(p\right) \cap R_{i} \cap \overline{R_{i}\left(p\right)}\right) \\
    &= R_{i}\left(p\right) \cup \left[R_{i+1}\left(p\right) \cap \left(\overline{R_{i}} \cup \left(R_{i} \cap \overline{R_{i}\left(p\right)}\right)\right)\right] \\
    &= \left[R_{i}\left(p\right) \cup R_{i+1}\left(p\right)\right] \cap \left[R_{i}\left(p\right) \cup \overline{R_{i}} \cup \left(R_{i} \cap \overline{R_{i}\left(p\right)}\right)\right] \\
    &= \left[R_{i}\left(p\right) \cup R_{i+1}\left(p\right)\right] \cap \left[\left(R_{i}\left(p\right) \cup \overline{R_{i}} \right) \cup \left(R_{i} \cap \overline{R_{i}\left(p\right)}\right)\right] \\
    &= \left[R_{i}\left(p\right) \cup R_{i+1}\left(p\right)\right] \cap \left[\left(R_{i} \cup \left(R_{i}\left(p\right) \cup \overline{R_{i}} \right) \right) \cap \left(\left(R_{i}\left(p\right) \cup \overline{R_{i}} \right) \cup \overline{R_{i}\left(p\right)}\right)\right] \\
    &= R_{i}\left(p\right) \cup R_{i+1}\left(p\right)
  \end{align*}
\end{proof}

\begin{lemma}
\label{lemma:multiple-step-progression-lemma}
For any steps $s_{0},s_{1}$, with $s_{1} > s_{0}$, and processor $p \in Procs$, \[\left[\bigcup_{i \in \left\{s_{0},\ldots,s_{1}\right\}} R_{i}\left(p\right)\right] \subseteq R_{s_{0}}\left(p\right) \cup \left[\bigcup_{i \in \left\{s_{0},\ldots,s_{1}-1\right\}} E_{i}\left(p\right) \cup M^{-}_{i}\left(p\right)\right].\]
\end{lemma}
\begin{proof}[Proof of Lemma~\ref{lemma:multiple-step-progression-lemma}]
  Prove this lemma by induction.
  \begin{description}
  \item[Base case] For the base case, let $s_{1} = s_{0} + 1$.
    Then, \[\left[\bigcup_{i \in \left\{s_{0},\ldots,s_{1}\right\}} R_{i}\left(p\right)\right] \subseteq R_{s_{0}}\left(p\right) \cup \left[\bigcup_{i \in \left\{s_{0},\ldots,s_{1}-1\right\}} E_{i}\left(p\right) \cup M^{-}_{i}\left(p\right)\right]\] iff $R_{s_{0}} \left(p\right) \cup R_{s_{1}} \left(p\right) \subseteq R_{s_{0}}\left(p\right) \cup E_{s_{0}}\left(p\right) \cup M^{-}_{s_{0}}\left(p\right)$.
    Taking into account Lemma~\ref{lemma:step-progression-lemma}, we conclude the base case holds.
  \item[Induction step] Assume that the result holds for some $s_{l} > s_{0}$, and show that it also holds for $s_{l}+1$.
    The induction hypothesis is
    \[\left[\bigcup_{i \in \left\{s_{0},\ldots,s_{l}\right\}} R_{i}\left(p\right)\right] \subseteq R_{s_{0}}\left(p\right) \cup \left[\bigcup_{i \in \left\{s_{0},\ldots,s_{l}-1\right\}} E_{i}\left(p\right) \cup M^{-}_{i}\left(p\right)\right]\]
    and prove
    \begin{align*}
      \left[\bigcup_{i \in \left\{s_{0},\ldots,s_{l}+1\right\}} R_{i}\left(p\right)\right] &\subseteq R_{s_{0}}\left(p\right) \cup \left[\bigcup_{i \in \left\{s_{0},\ldots,\left(s_{l}+1\right)-1\right\}} E_{i}\left(p\right) \cup M^{-}_{i}\left(p\right)\right] \\
      &= R_{s_{0}}\left(p\right) \cup \left[\bigcup_{i \in \left\{s_{0},\ldots,s_{l}-1\right\}} E_{i}\left(p\right) \cup M^{-}_{i}\left(p\right)\right] \\
      &\quad\cup E_{s_{l}}\left(p\right) \cup M^{-}_{s_{l}}\left(p\right)\\
      &\supseteq \left[\bigcup_{i \in \left\{s_{0},\ldots,s_{l}\right\}} R_{i}\left(p\right)\right] \cup E_{s_{l}}\left(p\right) \cup M^{-}_{s_{l}}\left(p\right)
    \end{align*}
    Again, taking into account Lemma~\ref{lemma:step-progression-lemma}, it is easy to deduce that the induction hypothesis holds, implying the lemma holds.
  \end{description}
\end{proof}

\begin{lemma}
\label{lemma:r-next-subseteq}
For any round $\overline{t}$ and processor $p \in Procs$
\begin{align*}
  R_{\overline{t+1}}\left(p\right) &\subseteq \left(E_{\overline{t}}\left(p\right) \cup R_{\overline{t}}\left(p\right) \cup M^{-}_{\overline{t}}\left(p\right)\right) - \left(C_{\overline{t}}\left(p\right) \cup M^{+}_{\overline{t}}\left(p\right)\right)
\end{align*}
\end{lemma}
\begin{proof}[Proof of Lemma~\ref{lemma:r-next-subseteq}]
  By Definitions~\ref{def:node-states}, \ref{def:node-definitions} and \ref{def:node-definitions-round}, it follows $R_{\overline{t+1}}\left(p\right) \cap C_{\overline{t}}\left(p\right) = \emptyset$.
  Since by Claim~\ref{claim:r-next-requirement-claim-1}, $R_{\overline{t+1}}\left(p\right) \cap M^{+}_{\overline{t}}\left(p\right) = \emptyset$, it follows $R_{\overline{t+1}}\left(p\right) \cap \left(C_{\overline{t}}\left(p\right) \cup M^{+}_{\overline{t}}\left(p\right)\right) = \emptyset$.
  Thus, it suffices to show that $R_{\overline{t+1}}\left(p\right) \subseteq E_{\overline{t}}\left(p\right) \cup R_{\overline{t}}\left(p\right) \cup M^{-}_{\overline{t}}\left(p\right)$.
  To conclude this proof, note that Lemma~\ref{lemma:multiple-step-progression-lemma}, with $s_{0} = \overline{t}\left[0\right]$ and $s_{1} = \overline{t+1}\left[0\right]$, implies just that.
\end{proof}

\begin{claim}
  \label{claim:claim:r-next-part-1-supseteq}
  For any round $\overline{t}$ and processor $p \in Procs$, $R_{\overline{t+1}}\left(p\right) \supseteq M^{-}_{\overline{t}}\left(p\right)- \left(C_{\overline{t}}\left(p\right) \cup M^{+}_{\overline{t}}\left(p\right)\right)$.
\end{claim}
\begin{proof}
  First, for an arbitrary step $s_{0}$ prove by induction that for any step $s_{1}$ such that $s_{1} > s_{0}$,
  \[R_{s_{1}}\parens{p} \supseteq \left[\bigcup_{i \in \left\{s_{0},\ldots,s_{1}-1\right\}} M^{-}_{i}\left(p\right)\right] - \left[\bigcup_{i \in \left\{s_{0},\ldots,s_{1}-1\right\}} C_{i}\left(p\right) \cup M^{+}_{i}\left(p\right)\right]\]
  \begin{description}
  \item[Base case] For the base case, let $s_{1} = s_{0} + 1$.
    Then, \[R_{s_{1}}\parens{p} \supseteq \left[\bigcup_{i \in \left\{s_{0},\ldots,s_{1}-1\right\}} M^{-}_{i}\left(p\right)\right] - \left[\bigcup_{i \in \left\{s_{0},\ldots,s_{1}-1\right\}} C_{i}\left(p\right) \cup M^{+}_{i}\left(p\right)\right]\] iff $R_{s_{1}}\parens{p} \supseteq M^{-}_{s_{0}}\left(p\right) - \left(C_{s_{0}}\left(p\right) \cup M^{+}_{s_{0}}\left(p\right)\right)$.
    To conclude the proof of the base case, note that
    \begin{align*}
      &M^{-}_{s_{0}}\left(p\right) - \left(C_{s_{0}}\left(p\right) \cup M^{+}_{s_{0}}\left(p\right)\right) \\
      &\subseteq M^{-}_{s_{0}}\left(p\right) \\
      &= R_{s_{1}}\left(p\right) \cap \left(R_{s_{0}} - R_{s_{0}}\left(p\right)\right)\\
      &\subseteq R_{s_{1}}\left(p\right)
    \end{align*}

  \item[Induction step] Assume that the result holds for some $s_{l} > s_{0}$, and show that it also holds for $s_{l}+1$.
    The induction hypothesis is
    \[R_{s_{l}}\parens{p} \supseteq \left[\bigcup_{i \in \left\{s_{0},\ldots,s_{l}-1\right\}} M^{-}_{i}\left(p\right)\right] - \left[\bigcup_{i \in \left\{s_{0},\ldots,s_{l}-1\right\}} C_{i}\left(p\right) \cup M^{+}_{i}\left(p\right)\right]\]
    and we prove
    \begin{align*}
      R_{s_{l+1}}\parens{p} &\supseteq \left[\bigcup_{i \in \left\{s_{0},\ldots,\left(s_{l}+1\right)-1\right\}} M^{-}_{i}\left(p\right)\right] - \left[\bigcup_{i \in \left\{s_{0},\ldots,\left(s_{l}+1\right)-1\right\}} C_{i}\left(p\right) \cup M^{+}_{i}\left(p\right)\right] \\
                            &= \left(M^{-}_{s_{l}} \left(p\right) \cup\left[\bigcup_{i \in \left\{s_{0},\ldots,s_{l}-1\right\}} M^{-}_{i}\left(p\right)\right]\right) \\
                            &\quad- \left[\bigcup_{i \in \left\{s_{0},\ldots,\left(s_{l}+1\right)-1\right\}} C_{i}\left(p\right) \cup M^{+}_{i}\left(p\right)\right] \\
                            &= \left(M^{-}_{s_{l}} \left(p\right) - \left[\bigcup_{i \in \left\{s_{0},\ldots,\left(s_{l}+1\right)-1\right\}} C_{i}\left(p\right) \cup M^{+}_{i}\left(p\right)\right]\right) \\
                            &\qquad \cup \left(\left[\bigcup_{i \in \left\{s_{0},\ldots,s_{l}-1\right\}} M^{-}_{i}\left(p\right)\right] - \left[\bigcup_{i \in \left\{s_{0},\ldots,\left(s_{l}+1\right)-1\right\}} C_{i}\left(p\right) \cup M^{+}_{i}\left(p\right)\right]\right) \\
                            &\subseteq M^{-}_{s_{l}} \left(p\right) \cup \left(\left[\bigcup_{i \in \left\{s_{0},\ldots,s_{l}-1\right\}} M^{-}_{i}\left(p\right)\right] - \left[\bigcup_{i \in \left\{s_{0},\ldots,\left(s_{l}+1\right)-1\right\}} C_{i}\left(p\right) \cup M^{+}_{i}\left(p\right)\right]\right)\\
      &= M^{-}_{s_{l}} \left(p\right) \cup \left(\left[\bigcup_{i \in \left\{s_{0},\ldots,s_{l}-1\right\}} M^{-}_{i}\left(p\right)\right]\right. \\
      &\qquad - \left.\left(\left[\bigcup_{i \in \left\{s_{0},\ldots,s_{l}-1\right\}} C_{i}\left(p\right) \cup M^{+}_{i}\left(p\right)\right] \cup C_{s_{l}}\left(p\right) \cup M^{+}_{s_{l}}\left(p\right)\right)\right)
      \end{align*}\begin{align*}
      &= M^{-}_{s_{l}} \left(p\right) \cup \left(\left(\left[\bigcup_{i \in \left\{s_{0},\ldots,s_{l}-1\right\}} M^{-}_{i}\left(p\right)\right]\right.\right. \\
      &\qquad - \left.\left.\left[\bigcup_{i \in \left\{s_{0},\ldots,s_{l}-1\right\}} C_{i}\left(p\right) \cup M^{+}_{i}\left(p\right)\right]\right) - \left(C_{s_{l}}\left(p\right) \cup M^{+}_{s_{l}}\left(p\right)\right)\right) \\
      &\subseteq M^{-}_{s_{l}} \left(p\right) \cup \left(R_{s_{l}}\left(p\right) - \left(C_{s_{l}}\left(p\right) \cup M^{+}_{s_{l}}\left(p\right)\right)\right) \\
      &= \left(R_{s_{l}+1}\left(p\right) \cap \left(R_{s_{l}} - R_{s_{l}}\left(p\right)\right)\right) \cup \left(R_{s_{l}}\left(p\right) - \left(C_{s_{l}}\left(p\right) \cup M^{+}_{s_{l}}\left(p\right)\right)\right) \\
      &\subseteq R_{s_{l}+1}\left(p\right) \cup \left(R_{s_{l}}\left(p\right) - \left(C_{s_{l}}\left(p\right) \cup M^{+}_{s_{l}}\left(p\right)\right)\right) \\
      &= R_{s_{l}+1}\left(p\right) \cup \left(R_{s_{l}}\left(p\right) \cap \overline{C_{s_{l}}\left(p\right)} \cap \overline{M^{+}_{s_{l}}\left(p\right)}\right) \\
      &= R_{s_{l}+1}\left(p\right) \cup \left(R_{s_{l}}\left(p\right) \cap \left(\overline{R_{s_{l}}\left(p\right) - R_{s_{l}+1}}\right) \cap \left(\overline{R_{s_{l}}\left(p\right) \cap \left(R_{s_{l}+1} - R_{s_{l}+1}\left(p\right)\right)}\right)\right) \\
      &= R_{s_{l}+1}\left(p\right) \cup \left(R_{s_{l}}\left(p\right) \cap \left(\overline{R_{s_{l}}\left(p\right)} \cup R_{s_{l}+1}\right) \cap \left(\overline{R_{s_{l}}\left(p\right)} \cup \overline{R_{s_{l}+1}} \cup R_{s_{l}+1}\left(p\right)\right)\right) \\
      &= R_{s_{l}+1}\left(p\right) \cup \left[\left(\left(R_{s_{l}}\left(p\right) \cap \overline{R_{s_{l}}\left(p\right)}\right)\cup\left(R_{s_{l}}\left(p\right) \cap R_{s_{l}+1}\right)\right) \cap \left(\overline{R_{s_{l}}\left(p\right)} \cup \overline{R_{s_{l}+1}} \cup R_{s_{l}+1}\left(p\right)\right)\right] \\
      &= R_{s_{l}+1}\left(p\right) \cup \left[R_{s_{l}}\left(p\right) \cap R_{s_{l}+1} \cap \left(\overline{R_{s_{l}}\left(p\right)} \cup \overline{R_{s_{l}+1}} \cup R_{s_{l}+1}\left(p\right)\right)\right] \\
      &= R_{s_{l}+1}\left(p\right) \cup \left[R_{s_{l}+1} \cap \left(\left(R_{s_{l}}\left(p\right) \cap \overline{R_{s_{l}}\left(p\right)}\right) \cup \left(R_{s_{l}}\left(p\right) \cap \overline{R_{s_{l}+1}}\right) \cup \left(R_{s_{l}}\left(p\right) \cap R_{s_{l}+1}\left(p\right)\right)\right)\right] \\
      &\subseteq R_{s_{l}+1}\left(p\right) \cup \left[R_{s_{l}+1} \cap \left(\left(R_{s_{l}}\left(p\right) \cap \overline{R_{s_{l}+1}}\right) \cup R_{s_{l}+1}\left(p\right)\right)\right] \\
      &\subseteq R_{s_{l}+1}\left(p\right) \cup \left[\left(R_{s_{l}+1} \cap R_{s_{l}}\left(p\right) \cap \overline{R_{s_{l}+1}}\right) \cup \left(R_{s_{l}+1} \cap R_{s_{l}+1}\left(p\right)\right)\right] \\
      &\subseteq R_{s_{l}+1}\left(p\right) \cup \left(R_{s_{l}+1} \cap R_{s_{l}+1}\left(p\right)\right) \\
      &= R_{s_{l}+1}\left(p\right)
    \end{align*}
  \end{description}
  To conclude this proof, let $s_{0} = \overline{t}\left[0\right]$ and $s_{1} = \overline{t+1}\left[0\right]$.
\end{proof}

\begin{claim}
  \label{claim:claim:r-next-part-2-supseteq}
  For any steps $s_{0},s_{1}$ such that $s_{1} > s_{0}$:
  \[R_{s_{1}}\left(p\right) \cup \left[\bigcup_{i \in \left\{s_{0},\ldots,s_{1}-1\right\}} C_{i}\left(p\right)\right] \supseteq \left[\bigcap_{i \in \left\{s_{0},\ldots,s_{1}-1\right\}} R_{s_{0}}\left(p\right) \cap \left(\overline{R_{i}\left(p\right)} \cup \overline{R_{i+1}} \cup R_{i+1}\left(p\right)\right)\right]\]
\end{claim}
\begin{proof}
  Prove this claim for an arbitrary $s_{0}$ by induction on $s_{1}$.
  \begin{description}
  \item[Base case] For the base case, consider $s_{1} = s_{0} + 1$. Then
    \begin{align*}
      R_{s_{1}}\left(p\right) \cup \left[\bigcup_{i \in \left\{s_{0},\ldots,s_{1}-1\right\}} C_{i}\left(p\right)\right] &\supseteq \left[\bigcap_{i \in \left\{s_{0},\ldots,s_{1}-1\right\}} R_{s_{0}}\left(p\right) \cap \left(\overline{R_{i}\left(p\right)} \cup \overline{R_{i+1}} \cup R_{i+1}\left(p\right)\right)\right] \\
                                                                                                                        &\text{iff} \\
      R_{s_{1}}\left(p\right) \cup C_{s_{0}}\left(p\right) &\supseteq R_{s_{0}}\left(p\right) \cap \left(\overline{R_{s_{0}}\left(p\right)} \cup \overline{R_{s_{0}+1}} \cup R_{s_{0}+1}\left(p\right)\right)
    \end{align*}
    To conclude the proof of the base case, note that
    \begin{align*}
      & R_{s_{0}}\left(p\right) \cap \left(\overline{R_{s_{0}}\left(p\right)} \cup \overline{R_{s_{0}+1}} \cup R_{s_{0}+1}\left(p\right)\right) \\
      &= \left(R_{s_{0}}\left(p\right) \cap \overline{R_{s_{0}}\left(p\right)} \right) \cup \left(R_{s_{0}}\left(p\right) \cap \overline{R_{s_{1}}}\right) \cup \left(R_{s_{0}}\left(p\right) \cap R_{s_{1}}\left(p\right)\right) \\
      &\subseteq C_{s_{0}}\left(p\right) \cup R_{s_{1}}\left(p\right) \\
    \end{align*}

  \item[Induction step] Assuming the claim holds for $s_{l} \geq s_{0} + 1$, prove the claim also holds for $s_{l} + 1$.
    Since by the induction hypothesis
    \[R_{s_{l}}\left(p\right) \cup \left[\bigcup_{i \in \left\{s_{0},\ldots,s_{l}-1\right\}} C_{i}\left(p\right)\right] \supseteq \left[\bigcap_{i \in \left\{s_{0},\ldots,s_{l}-1\right\}} R_{s_{0}}\left(p\right) \cap \left(\overline{R_{i}\left(p\right)} \cup \overline{R_{i+1}} \cup R_{i+1}\left(p\right)\right)\right]\]
    it suffices to show that
    \begin{align*}
      R_{s_{l}+1}\left(p\right) \cup \left[\bigcup_{i \in \left\{s_{0},\ldots,\left(s_{l}+1\right)-1\right\}} C_{i}\left(p\right)\right] &\supseteq \left(R_{s_{l}}\left(p\right) \cup \left[\bigcup_{i \in \left\{s_{0},\ldots,s_{l}-1\right\}} C_{i}\left(p\right)\right]\right) \\
      &\cap \left[R_{s_{0}}\left(p\right) \cap \left(\overline{R_{s_{l}}\left(p\right)} \cup \overline{R_{s_{l}+1}} \cup R_{s_{l}+1}\left(p\right)\right)\right]
    \end{align*}
    To conclude,
    \begin{align*}
      &\parens[\Bigg]{R_{s_{l}}\left(p\right) \cup \bracks[\Bigg]{\bigcup_{i \in \left\{s_{0},\ldots,s_{l}-1\right\}} C_{i}\left(p\right)}} \cap \left[R_{s_{0}}\left(p\right) \cap \left(\overline{R_{s_{l}}\left(p\right)} \cup \overline{R_{s_{l}+1}} \cup R_{s_{l}+1}\left(p\right)\right)\right] \\
      &\subseteq\parens[\Bigg]{R_{s_{l}}\left(p\right) \cup \bracks[\Bigg]{\bigcup_{i \in \left\{s_{0},\ldots,s_{l}-1\right\}} C_{i}\left(p\right)}} \cap \left(\overline{R_{s_{l}}\left(p\right)} \cup \overline{R_{s_{l}+1}} \cup R_{s_{l}+1}\left(p\right)\right) \\
      &= \left(  \bracks[\Bigg]{\bigcup_{i \in \left\{s_{0},\ldots,s_{l}-1\right\}} C_{i}\left(p\right)} \cap \left(\overline{R_{s_{l}}\left(p\right)} \cup \overline{R_{s_{l}+1}} \cup R_{s_{l}+1}\left(p\right)\right)\right) \\
      &\qquad \cup \left(\left(R_{s_{l}}\left(p\right) \cap \overline{R_{s_{l}}\left(p\right)} \right) \cup \left(R_{s_{l}}\left(p\right) \cap \overline{R_{s_{l}+1}} \right) \cup \left(R_{s_{l}}\left(p\right) \cap R_{s_{l}+1}\left(p\right)\right)\right) \\
      &\subseteq \bracks[\Bigg]{\bigcup_{i \in \left\{s_{0},\ldots,s_{l}-1\right\}} C_{i}\left(p\right)} \cup C_{s_{l}}\left(p\right) \cup R_{s_{l}+1}\left(p\right) \\
      &\subseteq R_{s_{l}+1}\left(p\right) \cup \bracks[\Bigg]{\bigcup_{i \in \left\{s_{0},\ldots,s_{l}-1\right\}} C_{i}\left(p\right)}
    \end{align*}
    
  \end{description}
\end{proof}

\begin{claim}
  \label{claim:claim:r-next-part-3-supseteq}
  For any steps $s_{0},s_{1}$ such that $s_{1} > s_{0}$,
  \[R_{s_{1}}\left(p\right) \cup \left[\bigcup_{i \in \left\{s_{0},\ldots,s_{1}-1\right\}} C_{i}\left(p\right)\right] \supseteq \left[\bigcup_{i \in \left\{s_{0},\ldots,s_{1}-1\right\}} E_{i}\left(p\right) \right] - \left[\bigcup_{i \in \left\{s_{0},\ldots,s_{1}-1\right\}} M^{+}_{i}\left(p\right) \right]\]
\end{claim}
\begin{proof}
  Prove this claim for an arbitrary $s_{0}$ by induction on $s_{1}$.
  \begin{description}
  \item[Base case] For the base case, consider $s_{1} = s_{0} + 1$. Then
    \begin{align*}
      R_{s_{1}}\left(p\right) \cup \left[\bigcup_{i \in \left\{s_{0},\ldots,s_{1}-1\right\}} C_{i}\left(p\right)\right] &\supseteq \left[\bigcup_{i \in \left\{s_{0},\ldots,s_{1}-1\right\}} E_{i}\left(p\right) \right] - \left[\bigcup_{i \in \left\{s_{0},\ldots,s_{1}-1\right\}} M^{+}_{i}\left(p\right) \right] \\
                                                                                                                        &\text{iff} \\
      R_{s_{1}}\left(p\right) \cup C_{s_{0}}\left(p\right) &\supseteq E_{s_{0}}\left(p\right) - M^{+}_{s_{0}}\left(p\right)
    \end{align*}
    To conclude the proof of the base case, note that by Definition~\ref{def:node-definitions}
    \begin{align*}
      &E_{s_{0}}\left(p\right) - M^{+}_{s_{0}}\left(p\right)\\
      &= \left(R_{s_{0}+1}\left(p\right) - R_{s_{0}}\right) - \left(R_{s_{0}}\left(p\right) \cap \left(R_{s_{0}+1} - R_{s_{0}+1}\left(p\right)\right)\right)\\
      &\subseteq R_{s_{1}}\left(p\right)
    \end{align*}

  \item[Induction step] Assuming the claim is true for $s_{l} \geq s_{0} + 1$ show that it holds for $s_{l} + 1$.
    Thus, using the induction hypothesis
    \[R_{s_{l}}\left(p\right) \cup \left[\bigcup_{i \in \left\{s_{0},\ldots,s_{l}-1\right\}} C_{i}\left(p\right)\right] \supseteq \left[\bigcup_{i \in \left\{s_{0},\ldots,s_{l}-1\right\}} E_{i}\left(p\right) \right] - \left[\bigcup_{i \in \left\{s_{0},\ldots,s_{l}-1\right\}} M^{+}_{i}\left(p\right) \right]\]
    show that
    \begin{align*}
      R_{s_{l}+1}\left(p\right) \cup \left[\bigcup_{i \in \left\{s_{0},\ldots,\left(s_{l}+1\right)-1\right\}} C_{i}\left(p\right)\right] \supseteq \left[\bigcup_{i \in \left\{s_{0},\ldots,\left(s_{l}+1\right)-1\right\}} E_{i}\left(p\right) \right] - \left[\bigcup_{i \in \left\{s_{0},\ldots,\left(s_{l}+1\right)-1\right\}} M^{+}_{i}\left(p\right) \right]
    \end{align*}
    It follows
    \begin{align*}
      &\left[\bigcup_{i \in \left\{s_{0},\ldots,\left(s_{l}+1\right)-1\right\}} E_{i}\left(p\right) \right] - \left[\bigcup_{i \in \left\{s_{0},\ldots,\left(s_{l}+1\right)-1\right\}} M^{+}_{i}\left(p\right) \right] \\
      &= \left[\bigcup_{i \in \left\{s_{0},\ldots,\left(s_{l}+1\right)-1\right\}} E_{i}\left(p\right) \right] \cap \overline{\left[\bigcup_{i \in \left\{s_{0},\ldots,\left(s_{l}+1\right)-1\right\}} M^{+}_{i}\left(p\right) \right]} \\
      &=\left[\bigcup_{i \in \left\{s_{0},\ldots,\left(s_{l}+1\right)-1\right\}} E_{i}\left(p\right) \right] \cap \left[\bigcap_{i \in \left\{s_{0},\ldots,\left(s_{l}+1\right)-1\right\}} \overline{M^{+}_{i}\left(p\right)} \right] \\
      &=\overline{M^{+}_{s_{l}}\left(p\right)} \cap \left(\left[\bigcup_{i \in \left\{s_{0},\ldots,\left(s_{l}+1\right)-1\right\}} E_{i}\left(p\right) \right] \cap \left[\bigcap_{i \in \left\{s_{0},\ldots,s_{l}-1\right\}} \overline{M^{+}_{i}\left(p\right)} \right]\right) \\
      &=\overline{M^{+}_{s_{l}}\left(p\right)} \cap \left(\left(E_{s_{l}}\left(p\right) \cup \left[\bigcup_{i \in \left\{s_{0},\ldots,s_{l}-1\right\}} E_{i}\left(p\right) \right] \right)\cap \left[\bigcap_{i \in \left\{s_{0},\ldots,s_{l}-1\right\}} \overline{M^{+}_{i}\left(p\right)} \right]\right) \\
      &=\overline{M^{+}_{s_{l}}\left(p\right)} \cap \left(\left(E_{s_{l}}\left(p\right) \cup \left[\bigcup_{i \in \left\{s_{0},\ldots,s_{l}-1\right\}} E_{i}\left(p\right) \right] \right)\cap \overline{\left[\bigcup_{i \in \left\{s_{0},\ldots,s_{l}-1\right\}} M^{+}_{i}\left(p\right) \right]}\right) \\
      &=\overline{M^{+}_{s_{l}}\left(p\right)} \cap \left[\left(E_{s_{l}}\left(p\right) \cap \overline{\left[\bigcup_{i \in \left\{s_{0},\ldots,s_{l}-1\right\}} M^{+}_{i}\left(p\right) \right]} \right) \right.\\
      &\left. \qquad \cup \left(\left[\bigcup_{i \in \left\{s_{0},\ldots,s_{l}-1\right\}} E_{i}\left(p\right) \right] \cap \overline{\left[\bigcup_{i \in \left\{s_{0},\ldots,s_{l}-1\right\}} M^{+}_{i}\left(p\right) \right]} \right)\right] \\
      &\subseteq \overline{M^{+}_{s_{l}}\left(p\right)} \cap \left[\left(E_{s_{l}}\left(p\right) \cap \overline{\left[\bigcup_{i \in \left\{s_{0},\ldots,s_{l}-1\right\}} M^{+}_{i}\left(p\right) \right]} \right) \cup \left(R_{s_{l}}\left(p\right) \cup \left[\bigcup_{i \in \left\{s_{0},\ldots,s_{l}-1\right\}} C_{i}\left(p\right)\right]\right)\right] \\
      &\subseteq \left(\overline{R_{s_{l}}\left(p\right)} \cup \overline{R_{s_{l}+1}} \cup R_{s_{l}+1}\left(p\right)\right) \\
      &\qquad \cap \left[\left(\left(R_{s_{l}+1}\left(p\right) - R_{s_{l}}\right) - \left[\bigcup_{i \in \left\{s_{0},\ldots,s_{l}-1\right\}} M^{+}_{i}\left(p\right) \right] \right) \cup \left(R_{s_{l}}\left(p\right) \cup \left[\bigcup_{i \in \left\{s_{0},\ldots,s_{l}-1\right\}} C_{i}\left(p\right)\right]\right)\right] \\
      \end{align*}
      \begin{align*}
      &\subseteq \left(\overline{R_{s_{l}}\left(p\right)} \cup \overline{R_{s_{l}+1}} \cup R_{s_{l}+1}\left(p\right)\right) \cap \left(R_{s_{l}+1}\left(p\right) \cup R_{s_{l}}\left(p\right) \cup \left[\bigcup_{i \in \left\{s_{0},\ldots,s_{l}-1\right\}} C_{i}\left(p\right)\right]\right) \\
      &= \left(R_{s_{l}}\left(p\right) \cup R_{s_{l}+1}\left(p\right) \cup \left[\bigcup_{i \in \left\{s_{0},\ldots,s_{l}-1\right\}} C_{i}\left(p\right)\right]\right) \cap \left(\overline{R_{s_{l}}\left(p\right)} \cup \overline{R_{s_{l}+1}} \cup R_{s_{l}+1}\left(p\right)\right) \\
      &= \left[\left(R_{s_{l}+1}\left(p\right) \cup \left[\bigcup_{i \in \left\{s_{0},\ldots,s_{l}-1\right\}} C_{i}\left(p\right)\right]\right) \cap \left(\overline{R_{s_{l}}\left(p\right)} \cup \overline{R_{s_{l}+1}} \cup R_{s_{l}+1}\left(p\right)\right)\right] \\
      &\qquad \cup \left[R_{s_{l}}\left(p\right) \cap \left(\overline{R_{s_{l}}\left(p\right)} \cup \overline{R_{s_{l}+1}} \cup R_{s_{l}+1}\left(p\right)\right)\right] \\
      &\subseteq R_{s_{l}+1}\left(p\right) \cup \left[\bigcup_{i \in \left\{s_{0},\ldots,s_{l}-1\right\}} C_{i}\left(p\right)\right] \\
      &\qquad \cup \left[\left(R_{s_{l}}\left(p\right) \cap \overline{R_{s_{l}}\left(p\right)}\right) \cup \left(R_{s_{l}}\left(p\right) \cap \overline{R_{s_{l}+1}}\right) \cup \left(R_{s_{l}}\left(p\right) \cap R_{s_{l}+1}\left(p\right)\right)\right] \\
      &\subseteq R_{s_{l}+1}\left(p\right) \cup \left[\bigcup_{i \in \left\{s_{0},\ldots,s_{l}-1\right\}} C_{i}\left(p\right)\right] \cup \left(R_{s_{l}}\left(p\right) - R_{s_{l}+1}\right) \cup R_{s_{l}+1}\left(p\right) \\
      &= R_{s_{l}+1}\left(p\right) \cup \left[\bigcup_{i \in \left\{s_{0},\ldots,s_{l}-1\right\}} C_{i}\left(p\right)\right] \cup C_{s_{l}}\left(p\right)\\
      &= R_{s_{l}+1}\left(p\right) \cup \left[\bigcup_{i \in \left\{s_{0},\ldots,\left(s_{l}+1\right)-1\right\}} C_{i}\left(p\right)\right]
    \end{align*}
  \end{description}
\end{proof}

\begin{lemma}
\label{lemma:r-next-supseteq}
For any round $\overline{t}$ and processor $p \in Procs$
\begin{align*}
  R_{\overline{t+1}}\left(p\right) &\supseteq \left(E_{\overline{t}}\left(p\right) \cup R_{\overline{t}}\left(p\right) \cup M^{-}_{\overline{t}}\left(p\right)\right) - \left(C_{\overline{t}}\left(p\right) \cup M^{+}_{\overline{t}}\left(p\right)\right)
\end{align*}
\end{lemma}
\begin{proof}[Proof of Lemma~\ref{lemma:r-next-supseteq}]
  To prove this direction of the lemma, it suffices to show:
  \begin{enumerate}
  \item $R_{\overline{t+1}}\left(p\right) \supseteq M^{-}_{\overline{t}}\left(p\right)- \left(C_{\overline{t}}\left(p\right) \cup M^{+}_{\overline{t}}\left(p\right)\right)$ \label{proof:r-next-supseteq-1}
  \item $R_{\overline{t+1}}\left(p\right) \supseteq R_{\overline{t}}\left(p\right) - \left(C_{\overline{t}}\left(p\right) \cup M^{+}_{\overline{t}}\left(p\right)\right)$ \label{proof:r-next-supseteq-2}
  \item $R_{\overline{t+1}}\left(p\right) \supseteq E_{\overline{t}}\left(p\right) - \left(C_{\overline{t}}\left(p\right) \cup M^{+}_{\overline{t}}\left(p\right)\right)$ \label{proof:r-next-supseteq-3}
  \end{enumerate}
  Prove each of these propositions:
  \begin{enumerate}
  \item Claim~\ref{claim:claim:r-next-part-1-supseteq} implies Proposition~\ref{proof:r-next-supseteq-1} holds.
  \item To prove Proposition~\ref{proof:r-next-supseteq-2}:
    \begin{align*}
      R_{\overline{t+1}}\left(p\right) &\supseteq R_{\overline{t}}\left(p\right) - \left(C_{\overline{t}}\left(p\right) \cup M^{+}_{\overline{t}}\left(p\right)\right) \\
                                       &= \left(R_{\overline{t}}\left(p\right) - M^{+}_{\overline{t}}\left(p\right)\right) - C_{\overline{t}}\left(p\right) \\
                                       &= \left(R_{\overline{t}\left[0\right]}\left(p\right) - \left[\bigcup_{i \in \left\{\overline{t}\left[0\right],\ldots,\overline{t}\left[L-1\right]\right\}}M^{+}_{i}\left(p\right)\right]\right) - C_{\overline{t}}\left(p\right) \\
                                       &= \left(R_{\overline{t}\left[0\right]}\left(p\right) - \left[\bigcup_{i \in \left\{\overline{t}\left[0\right],\ldots,\overline{t}\left[L-1\right]\right\}}R_{i}\left(p\right) \cap \left(R_{i+1} - R_{i+1}\left(p\right)\right)\right]\right) - C_{\overline{t}}\left(p\right) \\
                                       &= \left(R_{\overline{t}\left[0\right]}\left(p\right) \cap \left[\bigcap_{i \in \left\{\overline{t}\left[0\right],\ldots,\overline{t}\left[L-1\right]\right\}}\overline{R_{i}\left(p\right) \cap \left(R_{i+1} - R_{i+1}\left(p\right)\right)}\right]\right) - C_{\overline{t}}\left(p\right) \\
                                       &= \left(R_{\overline{t}\left[0\right]}\left(p\right) \cap \left[\bigcap_{i \in \left\{\overline{t}\left[0\right],\ldots,\overline{t}\left[L-1\right]\right\}}\overline{R_{i}\left(p\right)} \cup \overline{R_{i+1}} \cup R_{i+1}\left(p\right)\right]\right) - C_{\overline{t}}\left(p\right) \\
                                       &= \left[\bigcap_{i \in \left\{\overline{t}\left[0\right],\ldots,\overline{t}\left[L-1\right]\right\}} R_{\overline{t}\left[0\right]}\left(p\right) \cap \left(\overline{R_{i}\left(p\right)} \cup \overline{R_{i+1}} \cup R_{i+1}\left(p\right)\right)\right] - C_{\overline{t}}\left(p\right) \\
    \end{align*}
    By Claim~\ref{claim:claim:r-next-part-2-supseteq}, letting $s_{0} = \overline{t}\left[0\right]$ and $s_{1} = \overline{t+1}\left[0\right]$, it follows \[R_{\overline{t+1}}\left(p\right) \supseteq \left(R_{\overline{t+1}}\left(p\right) \cup C_{\overline{t}}\left(p\right)\right) - C_{\overline{t}}\left(p\right)\supseteq \left(R_{\overline{t}}\left(p\right) - M^{+}_{\overline{t}}\left(p\right)\right) - C_{\overline{t}}\left(p\right).\]

  \item By Claim~\ref{claim:claim:r-next-part-3-supseteq}, letting $s_{0} = \overline{t}\left[0\right]$ and $s_{1} = \overline{t+1}\left[0\right]$, it follows
    \begin{align*}
      R_{\overline{t+1}}\left(p\right) \cup C_{\overline{t}}\left(p\right) &\supseteq E_{\overline{t}}\left(p\right) - M^{+}_{\overline{t}}\left(p\right)
    \end{align*}
    To conclude this proof note that
    \begin{align*}
      R_{\overline{t+1}}\left(p\right) \supseteq \left(R_{\overline{t+1}}\left(p\right) \cup C_{\overline{t}}\left(p\right)\right) -  C_{\overline{t}}\left(p\right) \supseteq \left(E_{\overline{t}}\left(p\right) - M^{+}_{\overline{t}}\left(p\right)\right) -  C_{\overline{t}}\left(p\right)
    \end{align*}
  \end{enumerate}
\end{proof}

\begin{proof}[Proof of Lemma~\ref{lemma:r-next}]
  Lemmas~\ref{lemma:r-next-subseteq} and \ref{lemma:r-next-supseteq} imply this result.
\end{proof}


\subsection{Full proof for Lemma~\ref{lemma:connecting-lemma} (Connecting Lemma)}
\label{sub:1-proof-lemma-4}
\begin{claim}
\label{claim:progression-connecting-lemma-1}
For any round $\overline{t}$ and $p \in Procs$, $E_{\overline{t}}\left(p\right) \cap R_{\overline{t}}\left(p\right) = \emptyset$.
\end{claim}
\begin{proof}
  Given an arbitrary step $s_{0}$, we prove by induction on a step $s_{1}$ (where $s_{1} > s_{0}$) that $\left[\bigcup_{i \in \left\{s_{0},\ldots,s_{1}-1\right\}}E_{i}\left(p\right)\right] \cap R_{s_{0}}\left(p\right) = \emptyset$.
  \begin{description}
  \item[Base case] Let $s_{1} = s_{0} + 1$.
    Then $\left[\bigcup_{i \in \left\{s_{0},\ldots,s_{1}-1\right\}}E_{i}\left(p\right)\right] \cap R_{s_{0}}\left(p\right) = \emptyset$ iff $E_{s_{0}}\left(p\right) \cap R_{s_{0}}\left(p\right) = \emptyset$.
    By Definition~\ref{def:node-definitions}, it follows
    \begin{align*}
      &E_{s_{0}}\left(p\right) \cap R_{s_{0}}\left(p\right) \\
      &= \left(R_{s_{0}+1}\parens{p} - R_{s_{0}}\right) \cap R_{s_{0}}\left(p\right) \\
      &= R_{s_{1}}\parens{p} \cap \overline{R_{s_{0}}} \cap R_{s_{0}}\left(p\right) \\
      &\subseteq \overline{R_{s_{0}}} \cap R_{s_{0}} \\
      &= \emptyset.
    \end{align*}
  \item[Induction step] To prove the induction step, assume the lemma holds for a step $s_{l} > s_{0}$ and then prove that it also holds for $s_{l+1}$.
    \begin{align*}
      &\left[\bigcup_{i \in \left\{s_{0},\ldots,\left(s_{l}+1\right)-1\right\}}E_{i}\left(p\right)\right] \cap R_{s_{0}}\left(p\right) \\
      &=\left(\left[\bigcup_{i \in \left\{s_{0},\ldots,s_{l}-1\right\}}E_{i}\left(p\right)\right] \cap R_{s_{0}}\left(p\right)\right) \cup \left(E_{s_{l}}\left(p\right) \cap R_{s_{0}}\left(p\right)\right) \\
      &= E_{s_{l}}\left(p\right) \cap R_{s_{0}}\left(p\right)
    \end{align*}
    Since $s_{l} > s_{0}$, by Definition~\ref{def:node-states} it follows $E_{s_{l}} \cap R_{s_{0}} = \emptyset$, implying $E_{s_{l}}\left(p\right) \cap R_{s_{0}}\left(p\right) = \emptyset$.
  \end{description}
  To conclude the proof, let $s_{0} = \overline{t}\left[0\right]$ and $s_{1} = \overline{t+1}\left[0\right]$.
\end{proof}

\begin{claim}
\label{claim:progression-connecting-lemma-2}
For any round $\overline{t}$ and $p \in Procs$, $\left(R_{\overline{t}}\left(p\right) \cup E_{\overline{t}}\left(p\right)\right) \cap M^{-}_{\overline{t}}\left(p\right) = \emptyset$.
\end{claim}
\begin{proof}
  For the purpose of contradiction, let us assume $\left(R_{\overline{t}}\left(p\right) \cup E_{\overline{t}}\left(p\right)\right) \cap M^{-}_{\overline{t}}\left(p\right) \neq \emptyset$.
  Then, there must be a step $j \in \left\{\overline{t}\left[0\right],\ldots,\overline{t}\left[L-1\right]\right\}$ such that $\left(R_{\overline{t}}\left(p\right) \cup E_{\overline{t}}\left(p\right)\right) \cap M^{-}_{j}\left(p\right) \neq \emptyset$.
  Thus, at least one of the following propositions has to hold:
  \begin{enumerate}
  \item $R_{\overline{t}}\left(p\right) \cap M^{-}_{j}\left(p\right) \neq \emptyset$; \label{case:progression-connecting-lemma-2-1}
  \item $E_{\overline{t}}\left(p\right) \cap M^{-}_{j}\left(p\right) \neq \emptyset$. \label{case:progression-connecting-lemma-2-2}
  \end{enumerate}

  To conclude the proof of this claim, we prove that none of the propositions holds, contradicting our hypothesis that $\left(R_{\overline{t}}\left(p\right) \cup E_{\overline{t}}\left(p\right)\right) \cap M^{-}_{\overline{t}}\left(p\right) \neq \emptyset$.
  \begin{description}
  \item[Contradiction for Proposition~\ref{case:progression-connecting-lemma-2-1}] Let $S = R_{\overline{t}}\left(p\right) \cap M^{-}_{j}\left(p\right)$.
    Then,
    \begin{align*}
      S &= R_{\overline{t}}\left(p\right) \cap M^{-}_{j}\left(p\right) \\
        &= R_{\overline{t}\left[0\right]}\left(p\right) \cap \left(R_{j+1}\left(p\right) \cap \left(R_{j} - R_{j}\left(p\right)\right)\right) \\
        &= R_{\overline{t}\left[0\right]}\left(p\right) \cap R_{j+1}\left(p\right) \cap R_{j} \cap \overline{R_{j}\left(p\right)}
    \end{align*}
    If $j$ were $\overline{t}\left[0\right]$, then $S = \emptyset$, and so, $j > \overline{t}\left[0\right]$.
    Consider any node $\mu \in S$.
    Since a node that is \textbf{ready} can only become \textbf{executed}, and a node in state \textbf{executed} does not change its state, it follows $\forall i \in \left\{\overline{t}\left[0\right],\ldots,j+1\right\}, \mu \in R_{i}$.
    Noting that $\mu \in R_{\overline{t}\left[0\right]}\left(p\right) \cap \overline{R_{j}\left(p\right)} \cap R_{j}$, then there must be a step $k \in \left\{\overline{t}\left[0\right],\ldots,j-1\right\}$ such that $\mu \in R_{k}\left(p\right) \cap \overline{R_{k+1}\left(p\right)}$.
    Because $\forall i \in \left\{\overline{t}\left[0\right],\ldots,j+1\right\}, \mu \in R_{i}$, it follows $\mu \in R_{k}\left(p\right) \cap \overline{R_{k+1}\left(p\right)} \cap R_{k+1}$.
    By Definition~\ref{def:node-definitions}, it follows $\mu \in M^{+}_{k}\left(p\right)$, implying $\mu \in M^{+}_{\overline{t}}\left(p\right)$.
    However, since $\mu \in M^{-}_{\overline{t}}\left(p\right)$ and $\mu \in M^{+}_{\overline{t}}\left(p\right)$, it follows $M^{-}_{\overline{t}}\left(p\right) \cap M^{+}_{\overline{t}}\left(p\right) \neq \emptyset$, which, together with Lemma~\ref{lemma:migrated-to-p-subset-migrated-from-all-but-p}, contradicts Definition~\ref{def:round} --- the definition of a round.

  \item[Contradiction for Proposition~\ref{case:progression-connecting-lemma-2-2}] If $E_{\overline{t}}\left(p\right) \cap M^{-}_{j}\left(p\right) \neq \emptyset$, then there must be a step $m \in \left\{\overline{t}\left[0\right],\ldots,\overline{t}\left[L-1\right]\right\}$ such that $E_{m}\left(p\right) \cap M^{-}_{j}\left(p\right) \neq \emptyset$.
    Let $S = E_{m}\left(p\right) \cap M^{-}_{j}\left(p\right)$.
    It follows
    \begin{align*}
      S &= E_{m}\left(p\right) \cap M^{-}_{j}\left(p\right) \\
        &= \left(R_{m+1}\left(p\right) - R_{m}\right) \cap \left(R_{j+1}\left(p\right) \cap \left(R_{j} - R_{j}\left(p\right)\right)\right) \\
        &= R_{m+1}\left(p\right) \cap \overline{R_{m}} \cap R_{j+1}\left(p\right) \cap R_{j} \cap \overline{R_{j}\left(p\right)} \\
    \end{align*}
    Consider any node $\mu \in S$.

    If a node is not in state \textbf{ready} at step $m$, then it is either in state \textbf{not ready} or \textbf{executed}.
    Thus, at step $m$, $\mu$ is either in state \textbf{not ready} or \textbf{executed}.
    Because at step $m+1$ $\mu$ is in state \textbf{ready}, and since a node that is in state \textbf{executed} does not change its state, we deduce that $\mu$ is in state \textbf{not ready} at step $m$.
    Definition~\ref{def:node-states} then implies that until step $m$ (including $m$), $\mu$ has been in state \textbf{not ready}.

    By Definition~\ref{def:node-definitions}, a node can only be migrated at some step $i$ if it is ready at step $i$.
    Since $\mu$ is migrated at step $j$, then it must be ready at that step, implying $m < j$.
    Furthermore, if $m = j - 1$, then $S = \emptyset$, and so, it follows $m \in \left\{\overline{t}\left[0\right],\ldots,j-2\right\}$.
    
    Since a node that is \textbf{ready} can only become \textbf{executed}, and a node in state \textbf{executed} does not change its state, $\mu \in S$ implies $\forall i \in \left\{m+1,\ldots,j+1\right\}, \mu \in R_{i}$.
    Moreover, as $\mu \in R_{m+1}\left(p\right) \cap \overline{R_{j}\left(p\right)}$ and $m < j - 1$, it follows that there is a step $k \in \left\{m+1,\ldots,j\right\}$ such that $\mu \in R_{k}\left(p\right) \cap \overline{R_{k+1}\left(p\right)}$.
    Because $\forall i \in \left\{m+1,\ldots,j+1\right\}, \mu \in R_{i}$, it follows $\mu \in R_{k}\left(p\right) \cap \overline{R_{k+1}\left(p\right)} \cap R_{k+1}$.
    
    By Definition~\ref{def:node-definitions}, it follows $\mu \in M^{+}_{k}\left(p\right)$, implying $\mu \in M^{+}_{\overline{t}}\left(p\right)$.
    However, since $\mu \in M^{-}_{\overline{t}}\left(p\right)$ and $\mu \in M^{+}_{\overline{t}}\left(p\right)$, it follows $M^{-}_{\overline{t}}\left(p\right) \cap M^{+}_{\overline{t}}\left(p\right) \neq \emptyset$, which, together with Lemma~\ref{lemma:migrated-to-p-subset-migrated-from-all-but-p}, contradicts Definition~\ref{def:round} --- the definition of rounds. 
  \end{description}
\end{proof}

\begin{claim}
\label{claim:progression-connecting-lemma-3}
For any round $\overline{t}$ and $p \in Procs$, $C_{\overline{t}}\left(p\right) \cap M^{+}_{\overline{t}}\left(p\right) = \emptyset$.
\end{claim}
\begin{proof}
  Let $s_{0} = \overline{t}\left[0\right]$.
  We prove by induction on a step $s_{1} \in \left\{\overline{t}\left[0\right]+1,\ldots,\overline{t}\left[L-1\right]+1\right\}$ that $\left[\bigcup_{i \in \left\{s_{0},\ldots,s_{1}-1\right\}}C_{i}\left(p\right)\right] \cap \left[\bigcup_{i \in \left\{s_{0},\ldots,s_{1}-1\right\}}M^{+}_{i}\left(p\right)\right] = \emptyset$.
  \begin{description}
  \item[Base case] Let $s_{1} = s_{0} + 1$.
    Then $\left[\bigcup_{i \in \left\{s_{0},\ldots,s_{1}-1\right\}}C_{i}\left(p\right)\right] \cap \left[\bigcup_{i \in \left\{s_{0},\ldots,s_{1}-1\right\}}M^{+}_{i}\left(p\right)\right] = \emptyset$ iff $C_{s_{0}}\left(p\right) \cap M^{+}_{s_{0}}\left(p\right) = \emptyset$.
    By Definition~\ref{def:node-definitions}, it follows
    \begin{align*}
      &C_{s_{0}}\left(p\right) \cap M^{+}_{s_{0}}\left(p\right) \\
      &= \left(R_{s_{0}}\left(p\right) - R_{s_{1}}\right) \cap \left(R_{s_{0}}\left(p\right) \cap \left(R_{s_{1}} - R_{s_{1}}\left(p\right)\right)\right) \\
      &= R_{s_{0}}\left(p\right) \cap \overline{R_{s_{1}}} \cap R_{s_{0}}\left(p\right) \cap R_{s_{1}} \cap \overline{R_{s_{1}}\left(p\right)} \\
      &= \emptyset.
    \end{align*}

  \item[Induction step] To prove the induction step, assume the lemma holds for a step $s_{l} > s_{0}$ and then prove that it also holds for $s_{l}+1$, where $(s_{l}+1) \in \left\{\overline{t}\left[0\right]+1,\ldots,\overline{t}\left[L-1\right]+1\right\}$.
    The induction hypothesis is
    \[\left[\bigcup_{i \in \left\{s_{0},\ldots,s_{l}-1\right\}}C_{i}\left(p\right)\right] \cap \left[\bigcup_{i \in \left\{s_{0},\ldots,s_{l}-1\right\}}M^{+}_{i}\left(p\right)\right] = \emptyset.\]
    Thus,
    \begin{align*}
      &\left[\bigcup_{i \in \left\{s_{0},\ldots,\left(s_{l}+1\right)-1\right\}}C_{i}\left(p\right)\right] \cap \left[\bigcup_{i \in \left\{s_{0},\ldots,\left(s_{l}+1\right)-1\right\}}M^{+}_{i}\left(p\right)\right] \\
      &= \left(C_{s_{l}}\left(p\right) \cup \left[\bigcup_{i \in \left\{s_{0},\ldots,s_{l}-1\right\}}C_{i}\left(p\right)\right]\right) \cap \left(M^{+}_{s_{l}}\left(p\right) \cup \left[\bigcup_{i \in \left\{s_{0},\ldots,s_{l}-1\right\}}M^{+}_{i}\left(p\right)\right]\right) \\
      &= \left(C_{s_{l}}\left(p\right) \cap \left(M^{+}_{s_{l}}\left(p\right) \cup \left[\bigcup_{i \in \left\{s_{0},\ldots,s_{l}-1\right\}}M^{+}_{i}\left(p\right)\right]\right)\right) \\
      &\qquad \cup \left(\left[\bigcup_{i \in \left\{s_{0},\ldots,s_{l}-1\right\}}C_{i}\left(p\right)\right] \cap \left(M^{+}_{s_{l}}\left(p\right) \cup \left[\bigcup_{i \in \left\{s_{0},\ldots,s_{l}-1\right\}}M^{+}_{i}\left(p\right)\right]\right)\right) \\
      &= \left(M^{+}_{s_{l}}\left(p\right) \cap C_{s_{l}}\left(p\right)\right) \cup \left(\left[\bigcup_{i \in \left\{s_{0},\ldots,s_{l}-1\right\}}M^{+}_{i}\left(p\right)\right] \cap C_{s_{l}}\left(p\right)\right) \\
      &\qquad \cup \left(M^{+}_{s_{l}}\left(p\right) \cap \left[\bigcup_{i \in \left\{s_{0},\ldots,s_{l}-1\right\}}C_{i}\left(p\right)\right]\right) \\
      &\qquad \cup \left(\left[\bigcup_{i \in \left\{s_{0},\ldots,s_{l}-1\right\}}C_{i}\left(p\right)\right] \cap \left[\bigcup_{i \in \left\{s_{0},\ldots,s_{l}-1\right\}}M^{+}_{i}\left(p\right)\right]\right) \\
      &= \left(M^{+}_{s_{l}}\left(p\right) \cap C_{s_{l}}\left(p\right)\right) \cup \left(\left[\bigcup_{i \in \left\{s_{0},\ldots,s_{l}-1\right\}}M^{+}_{i}\left(p\right)\right] \cap C_{s_{l}}\left(p\right)\right) \\
      &\qquad \cup \left(M^{+}_{s_{l}}\left(p\right) \cap \left[\bigcup_{i \in \left\{s_{0},\ldots,s_{l}-1\right\}}C_{i}\left(p\right)\right]\right) \\
      &= \left(\left(R_{s_{l}}\left(p\right) - R_{s_{l}+1}\right) \cap \left(R_{s_{l}}\left(p\right) \cap \left(R_{s_{l}+1} - R_{s_{l}+1}\left(p\right)\right)\right)\right) \cup \left(\left[\bigcup_{i \in \left\{s_{0},\ldots,s_{l}-1\right\}}M^{+}_{i}\left(p\right)\right] \cap C_{s_{l}}\left(p\right)\right) \\
      &\qquad \cup \left(M^{+}_{s_{l}}\left(p\right) \cap \left[\bigcup_{i \in \left\{s_{0},\ldots,s_{l}-1\right\}}C_{i}\left(p\right)\right]\right) \\
      &= \left(R_{s_{l}}\left(p\right) \cap \overline{R_{s_{l}+1}} \cap R_{s_{l}}\left(p\right) \cap R_{s_{l}+1} \cap \overline{R_{s_{l}+1}\left(p\right)}\right) \cup \left(\left[\bigcup_{i \in \left\{s_{0},\ldots,s_{l}-1\right\}}M^{+}_{i}\left(p\right)\right] \cap C_{s_{l}}\left(p\right)\right)\\
      &\qquad \cup \left(M^{+}_{s_{l}}\left(p\right) \cap \left[\bigcup_{i \in \left\{s_{0},\ldots,s_{l}-1\right\}}C_{i}\left(p\right)\right]\right) \\
      &= \left(\left[\bigcup_{i \in \left\{s_{0},\ldots,s_{l}-1\right\}}M^{+}_{i}\left(p\right)\right] \cap C_{s_{l}}\left(p\right)\right) \cup \left(M^{+}_{s_{l}}\left(p\right) \cap \left[\bigcup_{i \in \left\{s_{0},\ldots,s_{l}-1\right\}}C_{i}\left(p\right)\right]\right) \\
    \end{align*}

    By Definition~\ref{def:node-definitions}, for any step $i$, $C_{i}\left(p\right)$ is composed by nodes that are attached to $p$ at step $i$ (implying they are \textbf{ready} at step $i$), but are no longer in state \textbf{ready} at step $i+1$.
    Thus, by Definition~\ref{def:node-states}, since a node that is in state \textbf{ready} can only change its state to \textbf{executed}, and since a node that is \textbf{executed} can not become \textbf{not ready} nor \textbf{ready}, it follows $\left[\bigcup_{i \in \left\{s_{0},\ldots,s_{l}-1\right\}}C_{i}\left(p\right)\right] \subseteq Executed_{s_{l}}$ (the set of nodes in state \textbf{executed} at step $s_{l}$).
    On the other hand, by Definition~\ref{def:node-definitions}, a node can only be migrated from $p$ at step $s_{l}$ if it is ready at that step, implying $M^{+}_{s_{l}}\left(Procs\right) \subseteq R_{s_{l}}$.
    With this, because a node can only be in one state at each step, it follows $Executed_{s_{l}} \cap R_{s_{l}} = \emptyset$, implying $\left[\bigcup_{i \in \left\{s_{0},\ldots,s_{l}-1\right\}}C_{i}\left(p\right)\right] \cap M^{+}_{s_{l}}\left(p\right) = \emptyset$.
    As such, to conclude this proof it only remains to show that $\left[\bigcup_{i \in \left\{s_{0},\ldots,s_{l}-1\right\}}M^{+}_{i}\left(p\right)\right] \cap C_{s_{l}}\left(p\right) = \emptyset$.

    Let $S = \left[\bigcup_{i \in \left\{s_{0},\ldots,s_{l}-1\right\}}M^{+}_{i}\left(p\right)\right] \cap C_{s_{l}}\left(p\right)$.
    For the purpose of contradiction, let us assume $S \neq \emptyset$.
    Thus, there is a step $j \in \left\{s_{0},\ldots,s_{l}-1\right\}$ such that $M^{+}_{j}\left(p\right) \cap C_{s_{l}}\left(p\right) \neq \emptyset$.
    Let $\mu$ be some node such that $\mu \in M^{+}_{j}\left(p\right) \cap C_{s_{l}}\left(p\right)$.
    By Definition~\ref{def:node-definitions}, it follows
    \begin{align*}
      &M^{+}_{j}\left(p\right) \cap C_{s_{l}}\left(p\right) \\
      &= \left(R_{j}\left(p\right) \cap \left(R_{j+1} - R_{j+1}\left(p\right)\right)\right) \cap \left(R_{s_{l}}\left(p\right) - R_{s_{l}+1}\right) \\
      &= R_{j}\left(p\right) \cap R_{j+1} \cap \overline{R_{j+1}\left(p\right)} \cap R_{s_{l}}\left(p\right) \cap \overline{R_{s_{l}+1}}
    \end{align*}
    which implies $\mu \in R_{j}\left(p\right) \cap R_{j+1} \cap \overline{R_{j+1}\left(p\right)} \cap R_{s_{l}}\left(p\right) \cap \overline{R_{s_{l}+1}}$.
    If $j$ were $s_{l} - 1$, it would follow $R_{j}\left(p\right) \cap R_{j+1} \cap \overline{R_{j+1}\left(p\right)} \cap R_{s_{l}}\left(p\right) \cap \overline{R_{s_{l}+1}} = \emptyset$, and so $j < s_{l} - 1$.
    Since a node that is \textbf{ready} can only become \textbf{executed}, and a node in state \textbf{executed} does not change its state, then $\forall i \in \left\{j,\ldots,s_{l}-1\right\}, \mu \in R_{i}$.
    Moreover, as $\mu \in \overline{R_{j+1}\left(p\right)} \cap R_{s_{l}}\left(p\right)$ and $s_{l} > j + 1$, it follows that there is a step $k \in \left\{j+1,\ldots,s_{l}-1\right\}$ such that $\mu \in R_{k+1}\left(p\right) \cap \overline{R_{k}\left(p\right)} \cap R_{k}$.
    By Definition~\ref{def:node-definitions}, it follows $\mu \in M^{-}_{k}\left(p\right)$, implying $\mu \in M^{-}_{\overline{t}}\left(p\right)$.
    However, since $\mu \in M^{-}_{\overline{t}}\left(p\right)$ and $\mu \in M^{+}_{\overline{t}}\left(p\right)$, it follows $M^{-}_{\overline{t}}\left(p\right) \cap M^{+}_{\overline{t}}\left(p\right) \neq \emptyset$, which, together with Lemma~\ref{lemma:migrated-to-p-subset-migrated-from-all-but-p}, contradicts Definition~\ref{def:round} (specifically, that no node is migrated more than once during the same round).
  \end{description}
\end{proof}

\begin{claim}
\label{claim:progression-connecting-lemma-4}
For any round $\overline{t}$ and $p \in Procs$, $R_{\overline{t}}\left(p\right) \cup E_{\overline{t}}\left(p\right) \cup M^{-}_{\overline{t}}\left(p\right) \supseteq C_{\overline{t}}\left(p\right) \cup M^{+}_{\overline{t}}\left(p\right)$.
\end{claim}
\begin{proof}
  By Requirement~\ref{requirement:node-execution-constraints}, it follows $R_{\overline{t}}\left(p\right) \cup E_{\overline{t}}\left(p\right) \cup M^{-}_{\overline{t}}\left(p\right) \supseteq C_{\overline{t}}\left(p\right)$.
  Thus, it suffices to show that $R_{\overline{t}}\left(p\right) \cup E_{\overline{t}}\left(p\right) \cup M^{-}_{\overline{t}}\left(p\right) \supseteq M^{+}_{\overline{t}}\left(p\right)$.
  By Definition~\ref{def:node-definitions}, for any step $i$, $M^{+}_{i}\left(p\right) = R_{i}\parens{p} \cap \parens{R_{i+1} - R_{i+1}\parens{p}}$, implying $R_{i}\left(p\right) \supseteq M^{+}_{i}\left(p\right)$.
  To conclude this proof, let $s_{0} = \overline{t}\left[0\right]$ and $s_{1} = \overline{t}\left[L-1\right]$ in Lemma~\ref{lemma:multiple-step-progression-lemma}.
\end{proof}

\begin{proof}[Proof of Lemma~\ref{lemma:connecting-lemma}]
  First, note that $\left|E_{\overline{t}}\left(p\right)\right| < \left|C_{\overline{t+1}}\left(p\right)\right| + \abs[\big]{M^{+}_{\overline{t}}\left(p\right)}$ iff
  \begin{align*}
    \left|E_{\overline{t}}\left(p\right)\right| + \left|R_{\overline{t}}\left(p\right)\right| + \abs[\big]{M^{-}_{\overline{t}}\left(p\right)} - \left|C_{\overline{t}}\left(p\right)\right| - \abs[\big]{M^{+}_{\overline{t}}\left(p\right)} &< \left|C_{\overline{t+1}}\left(p\right)\right| + \abs[\big]{M^{+}_{\overline{t}}\left(p\right)} + \left|R_{\overline{t}}\left(p\right)\right| + \abs[\big]{M^{-}_{\overline{t}}\left(p\right)} \\
                                                                                                                                                                                                                                              &\quad\, - \left|C_{\overline{t}}\left(p\right)\right| - \abs[\big]{M^{+}_{\overline{t}}\left(p\right)} \\
                                                                                                                                                                                                                                              &= \left|C_{\overline{t+1}}\left(p\right)\right| + \left|R_{\overline{t}}\left(p\right)\right| + \abs[\big]{M^{-}_{\overline{t}}\left(p\right)} - \left|C_{\overline{t}}\left(p\right)\right|
  \end{align*}
  
  Noting that:
  \begin{enumerate}
  \item Claim~\ref{claim:progression-connecting-lemma-1} implies \[\left|E_{\overline{t}}\left(p\right)\right| + \left|R_{\overline{t}}\left(p\right)\right| = \left|E_{\overline{t}}\left(p\right) \cup R_{\overline{t}}\left(p\right)\right|;\]
  \item Claim~\ref{claim:progression-connecting-lemma-2} implies \[\left|R_{\overline{t}}\left(p\right) \cup E_{\overline{t}}\left(p\right)\right| + \left|M^{-}_{\overline{t}}\left(p\right)\right| = \left|R_{\overline{t}}\left(p\right) \cup E_{\overline{t}}\left(p\right) \cup M^{-}_{\overline{t}}\left(p\right)\right|;\]
  \item Claim~\ref{claim:progression-connecting-lemma-3} implies \[\left|C_{\overline{t}}\left(p\right)\right| + \left|M^{+}_{\overline{t}}\left(p\right)\right| = \left|C_{\overline{t}}\left(p\right) \cup M^{+}_{\overline{t}}\left(p\right)\right|;\]
  \item Claim~\ref{claim:progression-connecting-lemma-4} implies
    \begin{align*}
      \left|E_{\overline{t}}\left(p\right) \cup R_{\overline{t}}\left(p\right) \cup M^{-}_{\overline{t}}\left(p\right)\right| &- \left|C_{\overline{t}}\left(p\right) \cup M^{+}_{\overline{t}}\left(p\right)\right| = \\
      &\left|\left(E_{\overline{t}}\left(p\right) \cup R_{\overline{t}}\left(p\right) \cup M^{-}_{\overline{t}}\left(p\right)\right) - \left(C_{\overline{t}}\left(p\right) \cup M^{+}_{\overline{t}}\left(p\right)\right)\right|; \text{and}
    \end{align*}
  \item Lemma~\ref{lemma:r-next} implies \[\left|R_{\overline{t+1}}\left(p\right)\right| = \left|\left(E_{\overline{t}}\left(p\right) \cup R_{\overline{t}}\left(p\right) \cup M^{-}_{\overline{t}}\left(p\right)\right) - \left(C_{\overline{t}}\left(p\right) \cup M^{+}_{\overline{t}}\left(p\right)\right)\right|,\]    
  \end{enumerate}
  it follows $\left|R_{\overline{t+1}}\left(p\right)\right| = \left|E_{\overline{t}}\left(p\right)\right| + \left|R_{\overline{t}}\left(p\right)\right| + \left|M^{-}_{\overline{t}}\left(p\right)\right| - \left|C_{\overline{t}}\left(p\right)\right| - \left|M^{+}_{\overline{t}}\left(p\right)\right|$.
  To conclude this proof, note that Requirement~\ref{requirement:node-execution-constraints} implies $\left|R_{\overline{t}}\left(p\right)\right| - \left|C_{\overline{t}}\left(p\right)\right| = \left|R_{\overline{t}}\left(p\right) - C_{\overline{t}}\left(p\right)\right|$ and $\left|R_{\overline{t+1}}\left(p\right)\right| - \left|C_{\overline{t+1}}\left(p\right)\right| = \left|R_{\overline{t+1}}\left(p\right) - C_{\overline{t+1}}\left(p\right)\right|$, and so, it follows
  \begin{align*}
     \left|R_{\overline{t+1}}\left(p\right)\right| &< \left|C_{\overline{t+1}}\left(p\right)\right| + \left|R_{\overline{t}}\left(p\right)\right| + \abs[\big]{M^{-}_{\overline{t}}\left(p\right)} - \left|C_{\overline{t}}\left(p\right)\right| \\
    &\text{iff} \\
    \left|R_{\overline{t+1}}\left(p\right) - C_{\overline{t+1}}\left(p\right)\right| &< \left|R_{\overline{t}}\left(p\right) - C_{\overline{t}}\left(p\right)\right| + \abs[\big]{M^{-}_{\overline{t}}\left(p\right)}.
  \end{align*}
\end{proof}

\section{The lock-free deque semantics}
\label{sec:deque-semantics}
In this section, we present the specification of the relaxed semantics associated with the lock free deque's implementation as given in~\cite{DBLP:journals/mst/AroraBP01}.
The deque implements three methods:
\begin{enumerate}
  \item \textsl{pushBottom} -- adds an item to the bottom of the deque and does not return.
  \item \textsl{popBottom} -- returns the bottom-most item from the deque, or \textsc{empty}, if there is no node.
  \item \textsl{popTop} -- attempts to return the topmost item from the deque, or \textsc{empty}, if there is no node.
    If the attempt succeeds, a node is returned. Otherwise, the special value \textsc{abort} is returned.
\end{enumerate}

The implementation is said to be \emph{constant-time} iff any invocation to each of the three methods takes at most a constant number of steps to return, implying the sequence of instructions composing the invocation has constant length.

An invocation to one of the deque's methods is defined by a 4-tuple establishing:
\begin{enumerate*}
  \item the method invoked;
  \item the invocation's beginning time;
  \item the invocation's completion time; and
  \item the return value, if it exists.
\end{enumerate*}

A set of invocations meets the \emph{relaxed semantics} iff there is a set of \emph{linearization times} for the corresponding non-aborting invocations for which:
\begin{enumerate*}
  \item every non-aborting invocation's linearization time lies within the initiation and completion times of the respective invocation;
  \item no two linearization times coincide;
  \item the return values for each non-aborting invocation are consistent with a serial execution of the methods in the order given by the linearization times of the non-aborting invocations; and
  \item for each aborted \textsl{popTop} invocation $x$ to a deque $d$, there is another invocation removing the topmost item from $d$ whose linearization time falls between the beginning and completion times of $x$'s invocation.
\end{enumerate*}

A set of invocations is said to be \emph{good} iff \textsl{pushBottom} and \textsl{popBottom} are never invoked concurrently.
The deque implementation presented in \cite{DBLP:journals/mst/AroraBP01} has been proven to satisfy the relaxed semantics on any good set of invocations.
Note that any set of invocations made during the execution of a computation scheduled by either \ws\ or \wss\ is good, as the \textsl{pushBottom} and \textsl{popBotom} methods are exclusively invoked by the (unique) owner of the deque.
Thus, throughout the paper we simply assume that the relaxed semantics are met.

\section{Full proofs for the results obtained in Section~4}
\label{sec:proofs-section-wss-analysis}

\subsection{Full proof for Lemma~\ref{lemma:lower-bounds-expected-spreads-p}}
\label{sub:5-proof-lemma-1}


First, we prove that the greater is the number of processors making steal, the smaller are the chances that $p$'s spread attempts succeeds.

\begin{lemma}
  \label{lemma:p-spread-attempt-probability-decreases-with-more-spreads}
  Let $spreads\left(p,\alpha,d\right)$ be a function corresponding to the expected number of nodes that $p$ spreads during any round, where the ratio of idle processors of the round is $\alpha$ and the number of processors enabling two nodes is $d$.
  Then, $spreads\left(p,\alpha,d\right) \geq spreads\left(p,\alpha,P\left(1-\alpha\right)\right)$.
\end{lemma}
\begin{proof}[Proof of Lemma~\ref{lemma:p-spread-attempt-probability-decreases-with-more-spreads}]
  If $p$ targets a processor whose \algid{state} flag is set to \textsc{working}, then its spread attempt fails.
  Thus, in this case $p$ would not spread a node, regardless of the number of processors that make a spread attempt.
  However, if $p$ targets a processor whose \algid{state} flag is set to \textsc{idle}, then its attempt has a chance to succeed.
  We now consider the two possible situations:
  \begin{description}[leftmargin=0cm,noitemsep,topsep=5pt,parsep=5pt,partopsep=0pt]
  \item[$d = P\left(1 - \alpha\right)$] --- In this case, $spreads\left(p,\alpha,d\right) = spreads\left(p,\alpha,P\left(1-\alpha\right)\right)$.
  \item[$d \neq P \left(1 - \alpha\right)$] --- By definition there are $P \left(1 - \alpha\right)$ busy processors, implying that $d \leq P \left(1 - \alpha\right)$.
    Thus, for this case we conclude $d < P \left(1 - \alpha\right)$.
    Now, suppose $p$ targets some processor $q$ whose \algid{state} flag is set to \textsc{idle}.
    Thanks to the synchronous environment we have artificially created, and assuming that any call to \algid{UniformlyRandomNumber} takes the same number of steps, then every processor executes the \algid{CAS} instruction --- whose success dictates the success of the spread attempt --- at the same step (line 30 of Algorithm~\ref{algo:wss}).
    Finally, as a consequence of our assumptions regarding the \algid{CAS} instruction (see the first paragraph of Section~\ref{sec:analyzing-wss}) and since processors target donees uniformly at random, the greater the number of spread attempts that target $q$ the smaller are the chances for $p$'s spread attempt to succeed, concluding the proof of this lemma.
  \end{description}
\end{proof}

\begin{lemma}
\label{lemma:bab-lower bounds on the expected non-empty bins for a fixed number of balls}
Suppose there are $B$ bins, each of which is painted either red or green, and let $B_{R}$ and $B_{G}$ denote the initial number of red and green bins, respectively.
Additionally, let $\alpha$ denote the initial ratio of red bins, meaning $\alpha = \frac{B_{R}}{B}$ and $B \left(1 - \alpha\right) = B_{G}$.
Now, suppose there are $B_{R}$ cubes and $B_{G}$ balls.
First, each cube is tossed, independently and uniformly at random, into the bins.
After tossing all the $B_{R}$ cubes, count the number of cubes that landed in green bins, and, for each such cube, a red bin is painted green.
After finishing all the paintings, each of the $B_{G}$ balls is tossed, independently and uniformly at random, into the bins.
\\
Let $Y$ denote the number of bins that are still red, with at least one ball.
Then, \[E\left[Y\right] \geq B \alpha^{2} \left(1 - \euler^{- \left(1 - \alpha\right)}\right).\]
\end{lemma}
\begin{proof}[Proof of Lemma~\ref{lemma:bab-lower bounds on the expected non-empty bins for a fixed number of balls}]
Let $C_{G hit}$ and $B_{R \mapsto G}$ be two random variables, corresponding, respectively, to the number of cubes that land in green bins and to the number of red bins that are painted green.
Then, $B_{R \mapsto G} = C_{G hit}$, and thus
\begin{align*}
  E\left[B_{R \mapsto G}\right] &= E\left[C_{G hit}\right] \\
                                &= B \alpha \left(1 - \alpha\right).
\end{align*}
Similarly to Lemma~\ref{lemma:bab-lower bounds on the probability that i is non-empty}, for a bin $b_{i}$ let $Y_{i}$ be an indicator variable, defined as \[Y_{i} = \left\{\begin{matrix}
 \,1 & \text{if at least one ball lands in } b_{i};\\
 \,0 & \text{otherwise.}
\end{matrix}\right.\]
Thus, the probability that none of the $B_{G}$ balls lands in $b_{i}$ is
\begin{align*}
  P\left\{Y_{i} = 0\right\} &= \left(1 - \frac{1}{B}\right)^{B\left(1 - \alpha\right)} \\
                            &\leq \euler^{- \left(1 - \alpha\right)}.
\end{align*}
Since the probability that no ball lands in $b_{i}$ is independent from the number of red bins painted green (\textit{i.e.}~$Y_{i}$ is independent from $B_{R \mapsto G}$), then, for any $m$,
\[P\left\{Y_{i} = 0 | B_{R \mapsto G} = m\right\} = P\left\{Y_{i} = 0\right\}, \quad \text{and,} \quad P\left\{Y_{i} = 1 | B_{R \mapsto G} = m\right\} = P\left\{Y_{i} = 1\right\}.\]
It follows
\begin{align*}
  E\left[Y_{i} | B_{R \mapsto G} = m\right] &= 0 . P\left\{Y_{i} = 0 | B_{R \mapsto G} = m\right\} + 1 . P\left\{Y_{i} = 1 | B_{R \mapsto G} = m\right\} \\
                                            &= P\left\{Y_{i} = 1\right\} \\
                                            &\geq 1 - \euler^{- \left(1 - \alpha\right)}.
\end{align*}
Consider \[Y = \sum_{i = 1}^{B_{R} - B_{R \mapsto G}} Y_{i},\] corresponding to the number of bins that are still red with at least one ball.
By the linearity of expectation, it follows
\begin{align*}
E\left[Y | B_{R \mapsto G} = m\right] &= E\left[Y_{1} + Y_{2} + \ldots + Y_{B_{R} - m} | B_{R \mapsto G} = m\right] \\
&= \sum_{i = 1}^{B_{R} - m} E\left[Y_{i} | B_{R \mapsto G} = m\right] \\
&\geq \sum_{i = 1}^{B_{R} - m} \left(1 - \euler^{- \left(1 - \alpha\right)}\right) \\
&= \left(B_{R} - m\right) \left(1 - \euler^{- \left(1 - \alpha\right)}\right). \\
\end{align*}
To conclude this proof, by the law of total expectation it follows
\begin{align*}
  E\left[Y\right] &= \sum_{m = 0}^{B_{R}} \left(E\left[Y | B_{R \mapsto G} = m\right] P\left\{B_{R \mapsto G} = m\right\}\right) \\
                  &\geq \sum_{m = 0}^{B_{R}} \left(\left(B_{R} - m\right) \left(1 - \euler^{- \left(1 - \alpha\right)}\right) P\left\{B_{R \mapsto G} = m\right\}\right) \\
                  &= \left(1 - \euler^{- \left(1 - \alpha\right)}\right) \sum_{m = 0}^{B_{R}} \left(\left(B_{R} - m\right) P\left\{B_{R \mapsto G} = m\right\}\right) \\
                  &= \left(1 - \euler^{- \left(1 - \alpha\right)}\right) \left(\sum_{m = 0}^{B_{R}} \left(B_{R} P\left\{B_{R \mapsto G} = m\right\}\right) - \sum_{m = 0}^{B_{R}} \left(m P\left\{B_{R \mapsto G} = m\right\}\right)\right) \\
                  &= \left(1 - \euler^{- \left(1 - \alpha\right)}\right) \left(B_{R} - E\left[B_{R \mapsto G}\right]\right) \\
                  &= \left(1 - \euler^{- \left(1 - \alpha\right)}\right) \left(B \alpha - B \left(1 - \alpha\right) \alpha\right) \\
                  &= B \alpha^{2} \left(1 - \euler^{- \left(1 - \alpha\right)}\right)
\end{align*}
\end{proof}


We now obtain lower bounds on the total number of spreads (or donations) made to processors during the second phase of some scheduling iteration, assuming that all busy processors make a spread attempt.

\begin{lemma}
  \label{lemma:lower-bounds-total-number-of-spreads}
  Consider any round $\overline{t}$ during a computation's execution, and let $B_{\overline{t}}$ be the set of processors that are busy during $\overline{t}$.
  If $\forall p \in B_{\overline{t}}, \left|E_{\overline{t}}\left(p\right)\right| = 2$, then $\mathrm{E}\bracks{|Spread^{+}_{\overline{t}}\left(B_{\overline{t}}\right)|} \geq P \alpha_{\overline{t}}^{2} \left(1 - \euler^{- \left(1 - \alpha_{\overline{t}}\right)}\right)$.  
\end{lemma}

\begin{proof}[Proof of Lemma~\ref{lemma:lower-bounds-total-number-of-spreads}]
  We prove this result by making an analogy with Lemma~\ref{lemma:bab-lower bounds on the expected non-empty bins for a fixed number of balls}:
  \begin{enumerate*}
  \item the number of bins $B$ corresponds to the number of processors $P$;
  \item the initial ratio of red and green bins correspond, respectively, to the ratio of idle and busy processors during the round;
  \item each cube toss corresponds to a steal attempt;
  \item each red bin that is painted green corresponds to a processor that was idle but whose steal attempt succeeded, and thus changed its \algid{state} flag to \textsc{working}; and
  \item each ball toss corresponds to a spread attempt.
  \end{enumerate*}
  Note that we can make this analogy because all steal attempts (and consequent \algid{state} flag updates) take place during the first phase of scheduling iterations while all spread attempts take place during the second phase of scheduling iterations. 
  Thus, $\mathrm{E}\bracks{|Spread^{+}_{\overline{t}}\left(B_{\overline{t}}\right)|} \geq P \alpha_{\overline{t}}^{2} \left(1 - \euler^{- \left(1 - \alpha_{\overline{t}}\right)}\right)$.
\end{proof}

\begin{proof}[Proof of Lemma~\ref{lemma:lower-bounds-expected-spreads-p}]
  By Lemma~\ref{lemma:p-spread-attempt-probability-decreases-with-more-spreads} it follows that $\mathrm{E}\bracks{|Spread^{+}_{\overline{t}}\left(p_{\overline{t}}\right)|}$ is the smallest if all busy processors enabled two nodes, and thus made a spread attempt.
  By Lemma~\ref{lemma:lower-bounds-total-number-of-spreads}, the expected number of nodes spread during a round such that all busy processors make a spread attempt is at least $P \alpha_{\overline{t}}^{2} \left(1 - \euler^{- \left(1 - \alpha_{\overline{t}}\right)}\right)$.
  Since, as we already noted, all processors have the same chances to make a successful spread attempt, and because each spread attempt may migrate at most one node, it follows that the expected number of nodes spread by each processor that makes a spread attempt is the same.
  Thus, since the expected number of nodes that $p$ spreads is the smallest if all processors make a spread attempt, then, letting $B_{\overline{t}}$ denote the set of processors that are busy during $\overline{t}$, it follows
  \[\mathrm{E}\bracks{|Spread^{+}_{\overline{t}}\left(p_{\overline{t}}\right)|} = \frac{\mathrm{E}\bracks{|Spread^{+}_{\overline{t}}\left(B_{\overline{t}}\right)|}}{P(1 - \alpha_{\overline{t}})} \geq \frac{\alpha_{\overline{t}}^{2}}{1 - \alpha_{\overline{t}}} \left(1 - \euler^{- \left(1 - \alpha_{\overline{t}}\right)}\right).\]
\end{proof}

\subsection{Full proof for Lemma~\ref{lemma:wss-nodes-stolen-spread-round-intersection}}
\label{sub:5-proof-lemma-2}

\begin{proof}[Proof of Lemma~\ref{lemma:wss-nodes-stolen-spread-round-intersection}]
  As proved in Claim~\ref{claim:progression-connecting-lemma-1}, $R_{\overline{t}}\left(p\right) \cap E_{\overline{t}}\left(p\right) = \emptyset$.
  To conclude the proof of this lemma, note that, from Definitions~\ref{def:nodes-stolen} and \ref{def:nodes-spread}, and by the definition of Algorithm~\ref{algo:wss} we have $Stolen^{+}_{\overline{t}}\left(p\right) \subseteq R_{\overline{t}}\left(p\right)$ and $Spread^{+}_{\overline{t}}\left(p\right) \subseteq E_{\overline{t}}\left(p\right)$ and so the lemma holds.  
\end{proof}

\subsection{Full proof for Lemma~\ref{lemma:wss-inequation-proof}}
\label{sub:5-proof-lemma-3}

\begin{claim}
\label{lemma:v-alpha increasing}
Let \[v\left(\alpha\right) = \frac{-2 + \euler^{\alpha-1}\left(2 + \alpha\left(4 - 5\alpha + \alpha^{3}\right)\right)}{{\left(\alpha - 1\right)}^{3}}.\]
Then, $\forall \alpha \in \left[0.7375; 1\right[  \qquad v\left(\alpha\right) \geq 0$.
\end{claim}
\begin{proof}
Let \[f\left(\alpha\right) = -2 + \euler^{\alpha-1}\left(2 + \alpha\left(4 - 5\alpha + \alpha^{3}\right)\right)\]
and \[g\left(\alpha\right) = {\left(\alpha - 1\right)}^{3}.\]
Thus, \[v\left(\alpha\right) = \frac{f\left(\alpha\right)}{g\left(\alpha\right)}.\]
Since \[\frac{\mathrm{d} f\left(\alpha\right)}{\mathrm{d} \alpha} =  \euler^{-1 + \alpha}{\left(1 - \alpha\right)}^{2}\left(6 + 6\alpha + \alpha^{2}\right)\]
it follows that $\forall \alpha \in \left[0.7375; 1\right[$, $f\left(\alpha\right)$ is non-decreasing. \\
Consequently, $\forall \alpha \in \left[0.7375; 1\right[$
\begin{align*}
f\left(\alpha\right) &\leq f\left(1\right) \\
&= -2 + \euler^{1-1}\left(2 + 1\left(4 - 5.1 + 1^{3}\right)\right) \\
&= 0 \\
\end{align*}
Since, $\forall \alpha \in \left[0.7375; 1\right[$
\[g\left(\alpha\right) = {\left(-1 + \alpha\right)}^{3} < 0\]
we have that, $\forall \alpha \in \left[0.7375; 1\right[$
\[v\left(\alpha\right) = \frac{f\left(\alpha\right)}{g\left(\alpha\right)} \geq 0,\]
concluding the proof of the claim.
\end{proof}

\begin{proof}[Proof of Lemma~\ref{lemma:wss-inequation-proof}]
\begin{align*}
1 < 1 - \euler^{-\alpha} + \frac{\alpha^{2}}{1 - \alpha} \left(1 - \euler^{-\left(1 - \alpha\right)}\right)
\end{align*}
iff
\begin{align*}
0 < - \euler^{-\alpha} + \frac{\alpha^{2}}{1 - \alpha} \left(1 - \euler^{- \left(1 - \alpha\right)}\right)
\end{align*}
Let \[s\left(\alpha\right) = - \euler^{-\alpha} + \frac{\alpha^{2}}{1 - \alpha} \left(1 - \euler^{- \left(1 - \alpha\right)}\right)\]
Then \[\frac{\mathrm{d} s\left(\alpha\right)}{\mathrm{d} \alpha} = -1 + \euler^{-\alpha} + \frac{1 + \euler^{-1+\alpha}\alpha\left(-2 + \alpha^{2}\right )}{{\left(-1 + \alpha \right)}^{2}} \]
Let $t\left(\alpha\right)$ be defined as the last two terms of $\frac{\mathrm{d} s\left(\alpha\right)}{\mathrm{d} \alpha}$:
\[t\left(\alpha\right) = \frac{1 + \euler^{-1+\alpha}\alpha\left(-2 + \alpha^{2}\right )}{{\left(-1 + \alpha \right)}^{2}}\]
To prove that $t\left(\alpha\right)$ is non-decreasing $\forall \alpha \in \left[0.7375; 1\right[$, we compute its derivative.
\[\frac{\mathrm{d} t\left(\alpha\right)}{\mathrm{d} \alpha} =  \frac{-2 + \euler^{\alpha-1}\left(2 + \alpha\left(4 - 5\alpha + \alpha^{3}\right)\right)}{{\left(- 1 + \alpha \right)}^{3}} \]
By Claim~\ref{lemma:v-alpha increasing}, $\forall \alpha \in \left[0.7375; 1\right[$ we have $\frac{\mathrm{d} t\left(\alpha\right)}{\mathrm{d} \alpha} \geq 0$,
meaning that $t\left(\alpha\right)$ is non-decreasing for that interval.\\
It follows that $\forall \alpha \in \left[0.7375; 1\right[$ we have
\begin{align*}
t\left(\alpha\right) &= \frac{1 + \euler^{-1+\alpha}\alpha\left( -2 + \alpha^{2} \right )}{{\left(- 1 + \alpha \right)}^{2}} \\
&\geq \frac{1 + \euler^{-1+0.7375}0.7375\left( -2 + 0.7375^{2} \right )}{{\left(-1 + 0.7375 \right)}^{2}} \\
&> 2.5
\end{align*}

Consequently,
\begin{align*}
\frac{\mathrm{d} s\left(\alpha\right)}{\mathrm{d} \alpha} &= -1 + \euler^{-\alpha} + \frac{1 + \euler^{-1+\alpha}\alpha\left( -2 + \alpha^{2} \right )}{{\left(- 1 + \alpha \right)}^{2}} \\
&= -1 + \euler^{-\alpha} + t\left(\alpha\right)\\
&> \euler^{-\alpha} + 2.5\\
&> 0
\end{align*}
Thus, $\forall \alpha \in \left[0.7375; 1\right[$, $s\left(\alpha\right)$ is strictly increasing.
To conclude this proof, it only remains to note that for that same interval we have
\begin{align*}
s\left(\alpha\right) &\geq s\left(0.7375\right) \\
&= -\euler^{-0.7375} + \frac{{0.7375}^{2}}{1 - 0.7375} \left(1 - \euler^{- \left(1 - {0.7375}\right)}\right) \\
&> 0.00006 \\
&> 0
\end{align*}
\end{proof}

\end{document}